\documentclass[11pt]{article}
\usepackage{amsmath,amssymb,amsthm}
\usepackage{graphicx}
\usepackage[margin=1in]{geometry}
\usepackage{hyperref}
\usepackage{booktabs}
\usepackage[utf8]{inputenc}
\usepackage[margin = 1in]{geometry}
\usepackage{xparse}
\usepackage{amsmath}
\usepackage{amssymb}
\usepackage{amsfonts}
\usepackage{diagbox}
\usepackage{color}
\usepackage{multicol}
\usepackage{epstopdf}
\usepackage{subcaption}
\usepackage{bm}
\usepackage{booktabs}
\newtheorem{lemma}{Lemma}

\newtheorem{proposition}{Proposition}[section]
\usepackage{amsmath}
\numberwithin{equation}{section}

\title{From Turing Patterns to Spike Self-Replication and Nucleation in a Growing-Domain Tulip Breaking Virus Model}
\author{Fanya Xia
  \thanks{Department of Mathematics, Harbin Institute of Technology, Weihai, Shandong, China. } \and Chunyi Gai \thanks{Department of Mathematics and Statistics,
University of Northern British Columbia,
Prince George, BC V2N 4Z9, Canada.
Corresponding author: chunyi.gai@unbc.ca} } \date{\today}

\begin{document}
	\maketitle
	
    \begin{abstract}
    A recently proposed reaction-diffusion model for the tulip breaking virus (TBV) has been shown to generate striped pigmentation patterns on growing tulip petals.
    In this paper, we investigate the formation of localized spike patterns in this model in the asymptotic limit of small viral diffusivity. It is known that two distinct bifurcation mechanisms can occur which generate new spikes as the domain half-length $L$ slowly increases: spike self-replication and spike nucleation. Self-replication occurs through a saddle-node bifurcation of the inner spike core problem, whereas nucleation is associated with a fold of the nonlinear outer problem. Using matched asymptotic expansions, we derive critical conditions for these mechanisms and construct phase diagrams that predict which mechanism occurs first as $L$ increases, or whether nucleation is absent. We further combine these results with a Turing instability analysis of the homogeneous steady state to obtain a phase diagram of pattern-forming mechanisms.  An interesting finding is that, for the TBV model, the Turing-unstable region lies entirely within the parameter regime where spike nucleation is possible. These results provide further insight into the mechanisms underlying the formation and evolution of pigmentation patterns in virus-infected tulips.
\end{abstract}
	
	\section{Introduction}
	
    Tulip breaking is a striking example of naturally occurring pigmentation patterning in plants \cite{valverde2012viruses}. While most tulip petals exhibit relatively uniform pigmentation, infection by tulip breaking virus (TBV) can produce characteristic striped and flame-like pigmentation patterns on the petals \cite{agoston2020genetic, lesnaw2000tulip}. These visually distinctive patterns have attracted considerable attention, but how viral activity ultimately leads to their formation is still not fully understood.

    The classical theory of diffusion-driven instability, first introduced by Turing \cite{turing}, provides a fundamental mechanism for the emergence of spatial patterns from an initially homogeneous state. Reaction-diffusion mechanisms have since been widely studied in biological pattern formation \cite{koch1994biological,kondo2021studies}. In many activator-inhibitor and activator-substrate systems, pattern formation is driven by local self-enhancement coupled with longer-range inhibition or resource depletion \cite{gierer1972theory,meinhardt2000pattern,schnakenberg1979simple}. Motivated by this mechanism, a recent reaction--diffusion TBV model was proposed to describe the formation of striped pigmentation patterns induced by tulip breaking virus \cite{wong2025tulip}. The model couples viral concentration, cellular resources required for viral replication, and anthocyanin production, with viral replication locally depleting the available substrate. These interactions generate an activator--substrate mechanism capable of producing spatially heterogeneous pigmentation patterns through a diffusion-driven instability and provide a starting point for a more detailed mathematical study of pattern formation in the TBV system.

    An additional feature of tulip development is that the petals continue to grow while the pigmentation pattern forms. It is therefore natural to consider the reaction-diffusion system on a time-dependent spatial domain rather than on a fixed interval. Pattern formation on growing domains has been studied extensively in reaction--diffusion systems \cite{crampin2002mode,madzvamuse2010stability,murray2003spatial}. In particular, Crampin et al.~\cite{crampin_1999} investigated diffusion-driven pattern formation on uniformly growing one-dimensional domains and subsequently extended this framework to spatially non-uniform growth \cite{crampin_2002}. Domain growth can alter pattern stability, induce transitions between spatial modes, and generate new pattern structures as the system evolves \cite{barrass2006mode,madzvamuse2010stability}.
   
    In this work, we are interested in the nonlinear regime where the pattern consists of well-separated localized spikes, whose existence and stability have been studied extensively in singularly perturbed reaction-diffusion systems \cite{kolokolnikov2005existence,kolokolnikov2009existence,tzou2013stability}. As the domain grows, additional spikes may emerge through two distinct mechanisms: an existing spike may split into two, referred to as \textit{spike self-replication}, or a new spike may form either at the boundary or in the region between existing peaks, referred to as \textit{spike nucleation}. These mechanisms provide a natural framework for studying how spatial patterns evolve as an infected tulip petal continues to grow. A previous study \cite{gai2025asymptotic} developed an asymptotic framework for analyzing spike self-replication and nucleation in reaction-diffusion systems on growing one-dimensional domains. Here, we investigate how these mechanisms arise in the context of the tulip breaking virus model and how they are connected to the underlying Turing pattern-forming mechanism.
    
	\subsection{The TBV Model}
    Building on the reaction-diffusion framework introduced in \cite{wong2025tulip}, we consider a modified TBV model with a constant source of viral particles. After nondimensionalization, the model takes the form
    \begin{subequations}\label{TBV_model_source}
    \begin{align}
    \frac{\partial T}{\partial t} &= \frac{\alpha}{1+V} - \delta T + \frac{\rho T^2}{k^2+T^2} \,, \quad - L < x < L \,,   
    \\
	\frac{\partial B}{\partial t} &= D_B \frac{\partial^2 B}{\partial X^2} + (1-B)B - \eta V^2 B \, , \quad  -L < X < L \, , \label{eq:B}
    \\
	\frac{\partial V}{\partial t} &= D_V \frac{\partial^2 V}{\partial X^2} + \frac{\beta}{1+T}V^2B - \mu V + a\, , \quad  -L < X < L \, ,  \label{eq:V}
    \end{align}
    \end{subequations}
    with Neumann boundary conditions $T_x(\pm L, t) = B_x(\pm L,t) =V_x(\pm L,t)=0$. Here, $T, B,$ and $V$ denote the anthocyanin concentration, the cellular resource available for viral replication, and the virus concentration, respectively. The parameters $\alpha$ and $\delta$ represent basal anthocyanin production and degradation, while $\rho$ and $k$ characterize the strength and saturation scale of its self-enhanced production. The parameter $\eta$ measures resource depletion by the virus, and $\beta$ and $\mu$ denote viral replication and decay rates. Moreover,  In this paper, we focus on the regime $D_v\ll D_B$, for which the relatively slow diffusion of the virus naturally leads to the formation of highly localized viral patterns. To capture this separation of diffusion scales, we introduce a small parameter $\varepsilon \ll 1$ and set $D_v=\varepsilon^2$, where $\varepsilon\ll1$.

    In the modified model, the constant source term $a>0$ provides a phenomenological description of continued viral input into the local tissue. Tulip breaking virus can be transmitted between plants by aphid vectors, with infection spreading during the growing season \cite{de2011non,hammond1989field}. The term $a$ therefore serves as a simple effective representation of sustained viral replenishment within the modeled region.

    The original model proposed in \cite{wong2025tulip} successfully reproduces striped pigmentation patterns and exhibits pattern splitting as the domain grows. This formulation corresponds to $a=0$. Here, we consider the extension $a>0$, which removes the virus-free equilibrium and modifies the pattern-forming behavior of the system. In particular, the source term allows localized patterns to arise through additional mechanisms beyond those present in the $a=0$ case.

    To account for the growth of the tulip petal, we prescribe the domain size using the growth law adopted in \cite{wong2025tulip},
    \begin{equation}\label{eq:Lt}
    L(t)=\frac{L_0 e^{\kappa \frac{t}{\tau}}}{1+\dfrac{1}{\xi_g}\left(e^{\kappa\frac{t}{\tau}}-1\right)}.
    \end{equation}
    Here, $L(t)$ denotes the half-length of the growing domain and $L_0=L(0)$ is the initial domain length. The parameter $\tau$ sets the characteristic time scale of petal growth, while $\kappa$ controls the growth rate. The parameter $\xi_g$ determines the overall increase in domain size, with $L(t)\to \xi_g L_0$ as $t\to\infty$.
    
    To formulate the system on a fixed domain, we follow the Lagrangian framework as in \cite{crampin_1999}. Let $X\in[-L(t),L(t)]$ denote the physical coordinate, introducing the transformation $x=\frac{X}{L(t)}$, we obtain the following rescaled system for the TBV model \eqref{TBV_model_source} 
    \begin{subequations}\label{TBV_model_Lagrangian}
    \begin{align}
    \frac{\partial T}{\partial t} &=
    \frac{\alpha}{1+V}-\delta T
    +\frac{\rho T^2}{k^2+T^2}
    -\frac{\dot L(t)}{L(t)}T \,, \quad - 1 < x < 1 \,,   \label{eq:T-grow}
    \\
	\frac{\partial B}{\partial t}&=
    D_{B_L}\frac{\partial^2 B}{\partial x^2}+(1-B)B-\eta V^2B-\frac{\dot L(t)}{L(t)}B\,, \quad - 1 < x < 1 \,, \label{eq:B-grow}
    \\
    \frac{\partial V}{\partial t}&=
    \varepsilon_L^2\frac{\partial^2 V}{\partial x^2}
    +\frac{\beta}{1+T}V^2B-\mu V+a
    -\frac{\dot L(t)}{L(t)}V, \quad - 1 < x < 1 \,, \
    \label{eq:V-grow}
    \end{align}
    \end{subequations}
    where $D_{B_L}\equiv \frac{D_B}{L^2}$, and $\varepsilon_L \equiv \frac{\varepsilon}{L}$. From \eqref{eq:Lt}, the relative growth rate is
    \begin{equation} 
    \frac{\dot L(t)}{L(t)}=\frac{\kappa}{\tau}\frac{\xi_g-1}{\xi_g+e^{\kappa\frac{t}{\tau}}-1}.
    \label{eq:Ldot-over-L}
    \end{equation}
    In the analysis below, we assume that petal growth is slow relative to the intrinsic reaction-diffusion dynamics, so that $\tau\gg 1$, and $\frac{\dot L(t)}{L(t)}\ll1$. The dilution terms are therefore higher-order on the fast pattern-forming time scale and will be neglected in the leading-order. The domain size $L(t)$, however, is retained as a slowly varying parameter, allowing the evolution of localized patterns to be studied quasi-statically as the petal grows. Therefore, we obtain
    \begin{subequations}\label{TBV_main}
    \begin{align}
    \frac{\partial T}{\partial t} &=
    \frac{\alpha}{1+V}-\delta T
    +\frac{\rho T^2}{k^2+T^2} \,, \quad - 1 < x < 1 \,;   \label{eq:T-main}
    \\
	\frac{\partial B}{\partial t}&=
    D_{B_L}\frac{\partial^2 B}{\partial x^2}+(1-B)B-\eta V^2B\,, \quad - 1 < x < 1 \,, B_x(\pm 1)=0\, ; \label{eq:B-main}
    \\
    \frac{\partial V}{\partial t}&= \varepsilon^2_{L}\frac{\partial^2 V}{\partial x^2}
    +\frac{\beta}{1+T}V^2B-\mu V+a, \quad - 1 < x < 1 \,, V_x(\pm 1)=0\, .
    \label{eq:V-main}
    \end{align}
    \end{subequations}
    
	\subsection{The main results}
	\label{sec:main-results}
    For the three-component TBV system \eqref{TBV_main}, we first analyze the onset of pattern formation before studying the nonlinear localized-pattern regime. We begin with the classical Turing instability, deriving the critical parameter values at which the homogeneous equilibrium loses stability to a stationary spatial mode and characterizing the corresponding parameter regimes for pattern formation. The three-component structure also allows, in principle, for an oscillatory diffusion-driven instability, a mechanism that does not arise in standard two-component reaction-diffusion systems. In the small-virus-diffusion regime, analytical and numerical results indicate that no oscillatory diffusion-driven instability occurs along the kinetically stable equilibrium branch considered here. Thus, the relevant diffusion-driven instability for the TBV model is the classical stationary Turing instability.

    We then focus on a localized-pattern regime and investigate how the resulting spike structures evolve as the domain slowly grows. Two distinct mechanisms for the generation of new spikes are observed. One is spike self-replication, whereby an existing spike splits into two spikes. The other is spike nucleation, whereby a new spike is generated either at the domain boundaries or from the quiescent background between existing spikes.

    Figure~\ref{fig:intro_spikes} illustrates these two mechanisms on a gradually growing domain. As $L(t)$ increases, the localized pattern evolves slowly until the domain reaches a critical length, at which the pattern undergoes a qualitative change. In the left panel, for sufficiently small source strength $a$, existing spikes undergo successive self-replication. In the right panel, for larger $a$, the existing spikes keep their shape while a new spike emerges from the quiescent background between neighboring spikes, corresponding to nucleation.

    To analyze these behaviors, we treat the domain length as a slowly varying bifurcation parameter and use matched asymptotic analysis, following \cite{gai2025asymptotic}. The localized-pattern problem is reduced to an inner core problem coupled to an outer algebraic system, from which we derive the thresholds for spike self-replication and nucleation.  Comparing these thresholds shows how the source strength $a$ plays a key role in selecting between the two mechanisms as the domain grows.

	\begin{figure}[h]
		\centering		\includegraphics[width=\textwidth]{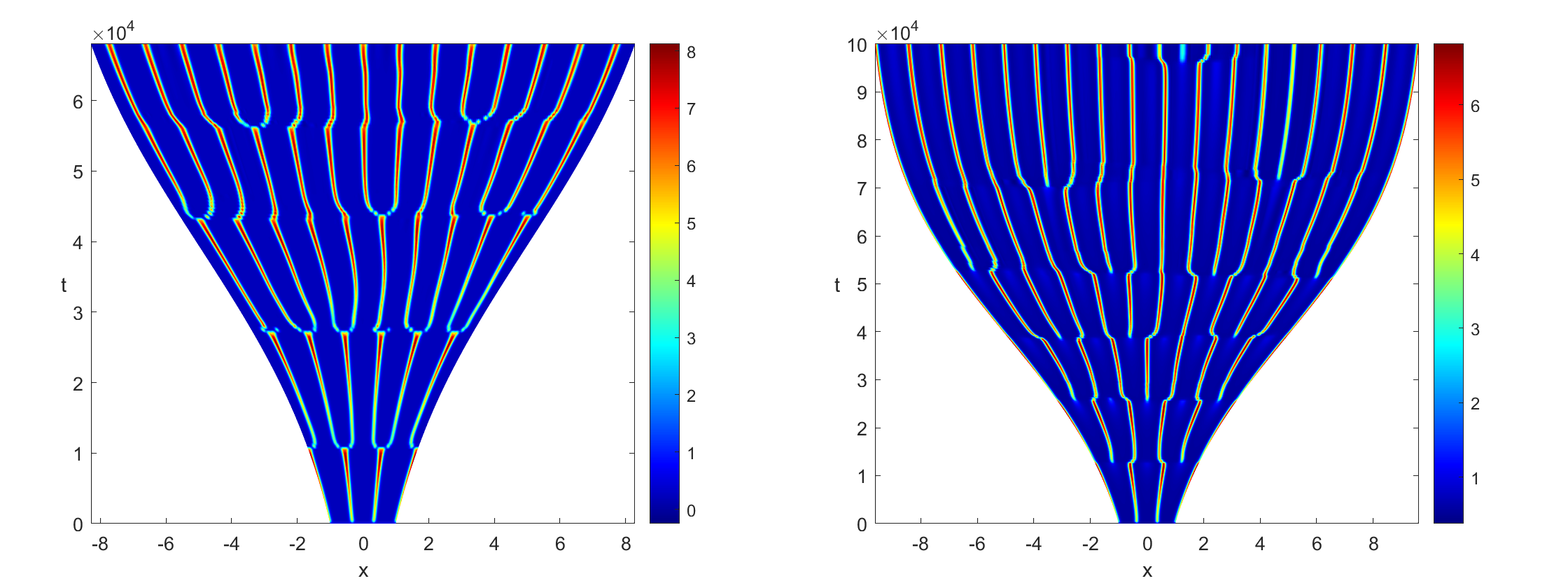}
		\caption{Full numerical simulations of the growing-domain TBV model \eqref{eq:T-grow}--\eqref{eq:V-grow}, computed using FlexPDE \cite{flexpde2015}, showing the space-time evolution of $V(x,t)$ on $x\in[-L(t),L(t)]$. Here, $L(t)$ is given by \eqref{eq:Lt} with $L_0=1$, $\omega=0.73$, $\xi_g=9.94$, and $\tau=1.3\times10^{4}$. Left: spike self-replication for $a=0.04$. Right: spike nucleation for $a=0.09$. The parameters are $\beta=1.5$, $\alpha=0.04$, $\delta=0.03$, $\rho=0.1$, $k=16$, $\eta=1$, $\mu=0.2$, $D_B=0.1$, and $D_V=1.3\times10^{-4}$.} 
		\label{fig:intro_spikes}
	\end{figure}

    The remainder of this paper is organized as follows. In Section \ref{Sec2:Turing}, we analyze diffusion-driven pattern formation in the three-component TBV system, including both classical and oscillatory Turing instabilities. In Section 3, we derive the inner core problem and determine the onset of spike self-replication. In Section 4, we construct the outer problem and determine the onset of spike nucleation. In Section 5, we compare the two mechanisms and characterize the parameter regimes in which spike self-replication or nucleation occurs as the domain grows. Finally, Section 6 summarizes the main findings and discusses possible extensions.
	

	\section{Turing pattern formation} 
	\label{Sec2:Turing}
    
    In this section, we analyze diffusion-driven instabilities of the full three-component TBV model \eqref{TBV_main}. We first compute the spatially homogeneous equilibria and their stability under the reaction kinetics, and then explore how diffusion modifies the spectrum of the linearized system. In particular, we consider two possible finite-wavenumber instabilities: a stationary Turing instability, associated with the loss of stability through a real eigenvalue, and an oscillatory Turing instability, associated with the crossing of a complex-conjugate pair. The analysis below determines the conditions under which each mechanism can occur in the TBV system.
    
    \subsection{Homogeneous equilibria and kinetic stability}
    \label{subsec:homo_eq}
    
    We first consider the spatially homogeneous dynamics associated with the TBV model \eqref{TBV_main}. Setting the diffusion terms to zero yields the spatially homogeneous kinetic system:   \begin{subequations}\label{TBV_ODE}
    \begin{align}
    \frac{dT}{dt}
    &= \frac{\alpha}{1+V}-\delta T+\frac{\rho T^2}{k^2+T^2},
    \label{ode_T}\\
    \frac{dB}{dt} 
    &=
    (1-B)B-\eta V^2B, \label{ode_B}\\
    \frac{dV}{dt}
    &=
    \frac{\beta}{1+T}V^2B-\mu V+a.
    \label{ode_V}
    \end{align}
    \end{subequations}
    
    Let $(\bar T,\bar B,\bar V)$ denote a homogeneous equilibrium of \eqref{TBV_ODE}. With the source term $a>0$, we consider steady states with $\bar V>0$. The equilibrium satisfies
    \begin{equation}
    \label{eq_eqn}
    \frac{\alpha}{1+\bar V}-\delta\bar T +\frac{\rho\bar T^2}{k^2+\bar T^2}=0, \qquad
    \bar B(1-\bar B-\eta\bar V^2)=0, \qquad
    \frac{\beta}{1+\bar T}\bar V^2\bar B-\mu\bar V+a=0.
    \end{equation}
    One equilibrium branch is obtained by setting $\bar B=0$.  In this case, the virus equation gives $\bar V=\frac{a}{\mu},$ while $\bar T$ satisfies
    \begin{equation}
    \bar T^3-b_1\bar T^2+b_2\bar T-b_3=0,
    \label{eq:Bzero_cubic}
    \end{equation}
    where $b_1= \frac{\rho}{\delta} +\frac{\alpha}{\delta(1+\bar V)}, b_2=k^2, b_3= \frac{\alpha k^2}{\delta(1+\bar V)}.$ Moreover, we conclude the following properties of the root
    \begin{proposition}
        Let  $\mathcal{P}(T)=T^3-b_1T^2+b_2T-b_3,$ with $b_1,b_2,b_3>0.$ Then $\mathcal{P}$ has no negative real roots and at least one positive real root. Moreover, if $b_1^2 \le 3b_2,$ then $\mathcal{P}$ has a unique positive real root. Otherwise, the number of real roots is determined by the discriminant
        \[\Delta =b_1^2b_2^2-4b_2^3-4b_1^3b_3-27b_3^2+18b_1b_2b_3.\]
        In particular, $\Delta>0$ gives three distinct positive real roots, whereas $\Delta<0$ gives one positive real root and a complex-conjugate pair. The case $\Delta=0$ corresponds to repeated positive real roots.
    \end{proposition}
    Since the subsequent pattern-forming analysis focuses on equilibria with a positive cellular resource, we now consider the positive-resource branch $\bar B>0$. From the second and third equation in \eqref{eq_eqn}, $\bar B=1-\eta\bar V^2,$ and $\mu \bar{V} -a>0$ hence a biologically admissible equilibrium requires
    \begin{equation}
    \frac{a}{\mu}<\bar V<\frac{1}{\sqrt{\eta}},
    \qquad 0<\bar B<1.
    \label{eq:V_admissible}
    \end{equation} 
    For a given $\bar V$, the first equation in \eqref{eq_eqn} determines $\bar T$ implicitly. Define
    \begin{equation}\label{fun_f}
        f(T,V):= \frac{\alpha}{1+V} -\delta T +\frac{\rho T^2}{k^2+T^2}.
    \end{equation}
    Along the positive equilibrium branch, we write
    \begin{equation}
    \bar T=\Theta(\bar V),
    \label{eqn:theta(v)}
    \end{equation}
    where $f(\Theta(\bar V),\bar V)=0.$ Provided $f_T\neq0$, the implicit function theorem gives
    \begin{equation}
    \Theta'(\bar V) = -\frac{f_V}{f_T} = \frac{\alpha}{(1+\bar V)^2f_T},
    \label{eq:Theta_prime}
    \end{equation}
    where
    \begin{equation}
    f_T=-\delta+\frac{2\rho k^2\bar T} {(k^2+\bar T^2)^2}.
    \label{eq:fT}
    \end{equation}
    The term $\frac{2\rho k^2T}{(k^2+T^2)^2}$ attains its maximum at $T=k/\sqrt{3}$. We therefore assume throughout that 
    \begin{equation}
    \delta>\frac{9\rho}{8\sqrt{3}\,k},
    \label{eq:fT_condition}
    \end{equation}
    which ensures that $f_T<0$ for all $T\geq0$. Since $f(0,\bar V)>0$ and $f(T,\bar V)\to-\infty$ as $T\to\infty$, it follows that, for each $\bar V$ in $(\frac{a}{\mu},\frac{1}{\sqrt{\eta}})$, there is a unique positive solution $\bar T=\Theta(\bar V)$. Moreover, $\Theta'(\bar V)<0$, so the equilibrium anthocyanin concentration decreases as the viral concentration increases.
    
    Now substituting $\bar T=\Theta(\bar V), \bar B=1-\eta\bar V^2$ into the virus equilibrium equation yields
    \begin{equation}
    \mathcal{Q}(\bar V;\mu,a) :=\frac{\beta}{1+\Theta(\bar V)} \bar V^2(1-\eta\bar V^2) -\mu\bar V+a=0. \label{eq:Q_equilibrium}
    \end{equation}
    Equivalently, solving \eqref{eq:Q_equilibrium} for the viral decay rate $\mu$ gives 
    \begin{equation}
    \mu=\mathcal M(\bar V;a):=
    \frac{\beta\bar V(1-\eta\bar V^2)}{1+\Theta(\bar V)} +\frac{a}{\bar V}, \qquad 0<\bar V<\frac{1}{\sqrt{\eta}}.
    \label{eq:mu_branch}
    \end{equation}
    Thus, the positive homogeneous equilibrium branch can be parametrized by $\bar V$ through
    \begin{equation}
    \bar T=\Theta(\bar V),
    \qquad
    \bar B=1-\eta\bar V^2, \qquad \mu=\mathcal M(\bar V;a).
    \label{eq:eq_param}
    \end{equation}
    At the endpoints of the admissible interval,
    \begin{equation}
    \mathcal M(\bar V;a)\to\infty \quad\text{as}\quad
    \bar V\to0^+,
    \qquad
    \mathcal M(\bar V;a)\to a\sqrt{\eta}
    \quad\text{as}\quad
    \bar V\to\frac{1}{\sqrt{\eta}}.
    \label{eq:M_limits}
    \end{equation}
    
    The representation \eqref{eq:mu_branch} also provides a convenient description of the bifurcation structure of the homogeneous equilibria. A fold point satisfies
    \begin{equation}
    \mathcal M_{\bar V}(\bar V;a)=0,
    \label{eq:fold_condition}
    \end{equation}
    and corresponds to a saddle-node bifurcation, provided $\mathcal M_{\bar V\bar V}(\bar V;a)\neq0,$. As the source strength $a$ is varied, two saddle-node points may coalesce at a cusp. This occurs when    
    \begin{equation}
    \mathcal M_{\bar V}=0,
    \qquad 
    \mathcal M_{\bar V\bar V}=0,
    \label{eq:cusp_condition}
    \end{equation}
    subject to the usual nondegeneracy and transversality conditions.
    
    We now illustrate this bifurcation structure using the representative parameter values listed in Table~\ref{tab:params}. Unless otherwise specified, the parameters other than $\mu$ and $a$ are fixed at these values throughout this section, while $\mu$ is used as the primary bifurcation parameter.
 	
	\begin{table}[h]
	\centering
	\begin{tabular}{ccccccccc}
		\toprule
		$\alpha$ & $\delta$ & $\rho$ & $k$ & $\eta$ & $\beta$ & $\mu$ & $D_B$ & $D_V$ \\
		\midrule
		0.04 & 0.03 & 0.1 & 16 & 2 & 1.5 & 0.2 & 0.013 & 0.00013 \\
		\bottomrule
	\end{tabular}
	\caption{Representative parameter values used in the bifurcation and stability analysis.}
	\label{tab:params}
    \end{table}

    For $a=0.02$, the equilibrium branch shown in Figure~\ref{fig:bifurcation} has an S-shaped structure with two saddle-node bifurcations located at $(\mu_1,\bar V_1)\approx(0.23,0.18),  (\mu_2,\bar V_2)\approx(0.26,0.39).$ Hence, for $\mu_1<\mu<\mu_2$, three positive homogeneous equilibria coexist in the ODE dynamics. The effect of varying the source strength $a$ is illustrated in Figure~\ref{fig:cusp_comparison}. As $a$ increases, the two saddle-node points move toward one another and eventually merge at a cusp. Solving the cusp conditions \eqref{eq:cusp_condition} numerically gives the critical value $a^*\approx0.037029.$  Thus, for $a<a^*$, the equilibrium branch contains two folds; at $a=a^*$, the two folds merge at a cusp point; and for $a>a^*$, the fold points disappear and the equilibrium branch becomes single-valued in $\mu$.
  \begin{figure}[h]
    \centering
    \begin{subfigure}{0.45\textwidth}
        \centering
        \includegraphics[width=\textwidth]{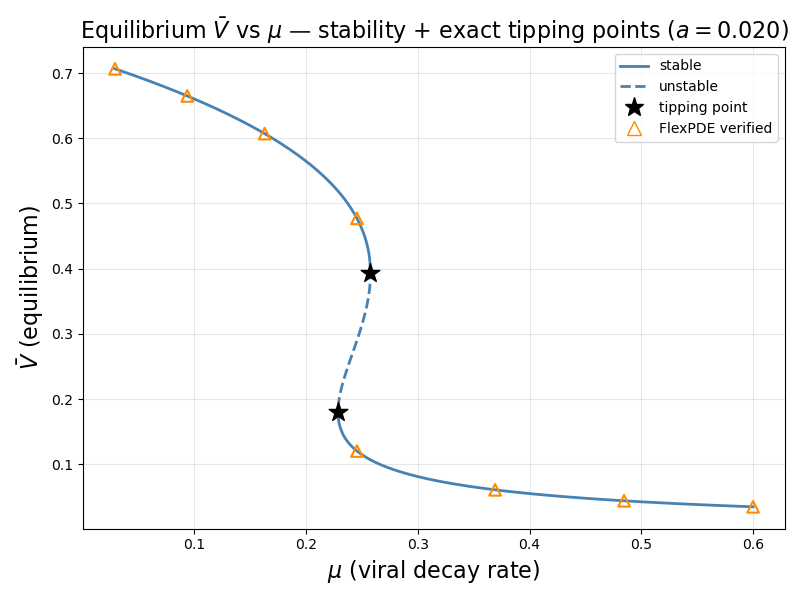}
        \caption{}
        \label{fig:bifurcation_a}
    \end{subfigure}
    \hfill
    \begin{subfigure}{0.45\textwidth}
        \centering
        \includegraphics[width=\textwidth]{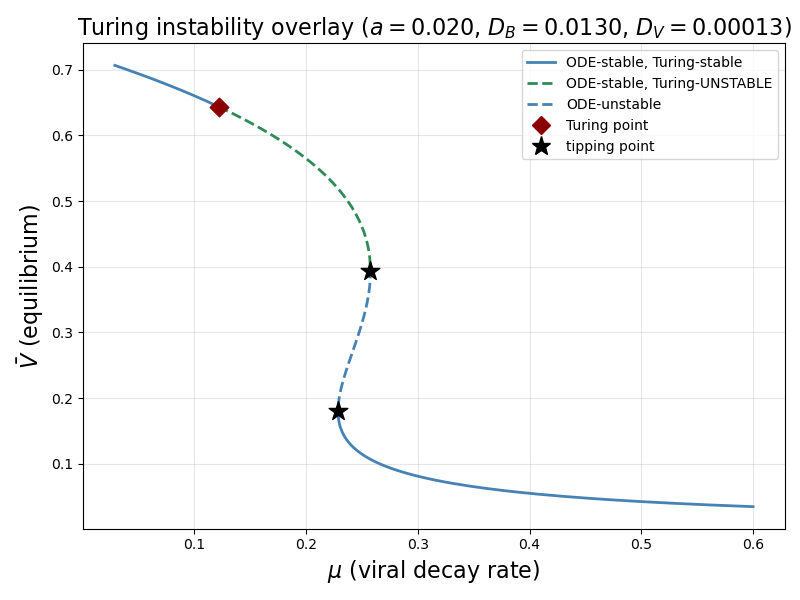}
        \caption{}
        \label{fig:bifurcation_b}
    \end{subfigure}
    \caption{(a) Bifurcation diagram of the homogeneous equilibria with respect to $\mu$ for $a=0.02$. Solid and dashed curves denote stable and unstable equilibria, respectively. The stars mark the two saddle-node bifurcation points, while the triangles indicate representative equilibria verified numerically using \textit{FlexPDE}. (b) Turing instability along the upper stable branch. The remaining parameters are chosen as listed in Table~\ref{tab:params}.}
    \label{fig:bifurcation}
    \end{figure}

    \begin{figure}[h]
    \centering
    \begin{subfigure}{0.325\textwidth}
        \centering
        \includegraphics[width=\textwidth]{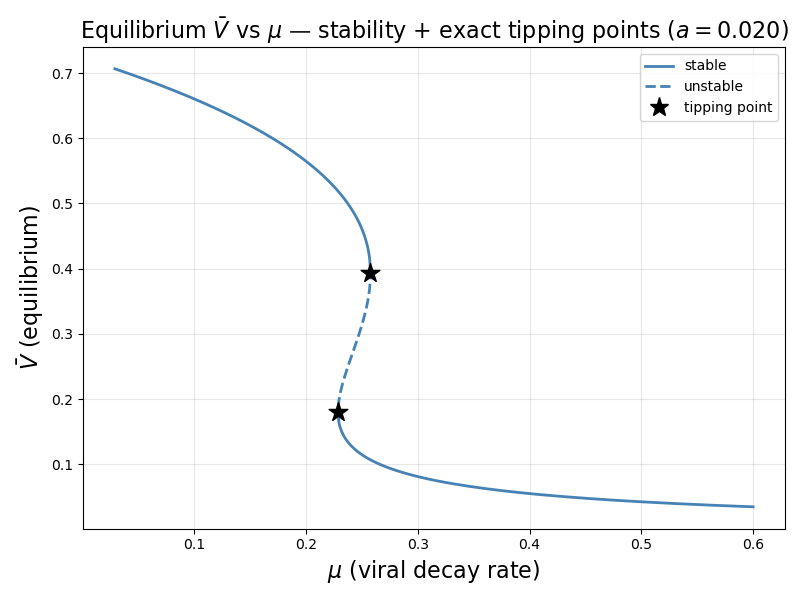}
        \caption{}
        \label{fig:cusp_a}
    \end{subfigure}
    \begin{subfigure}{0.325\textwidth}
        \centering
        \includegraphics[width=\textwidth]{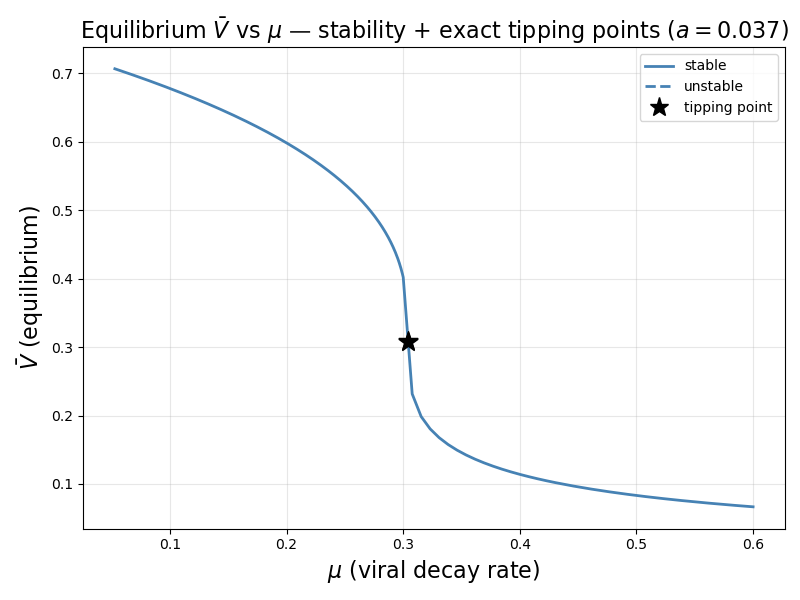}
        \caption{}
        \label{fig:cusp_b}
    \end{subfigure}
    \begin{subfigure}{0.325\textwidth}
        \centering
        \includegraphics[width=\textwidth]{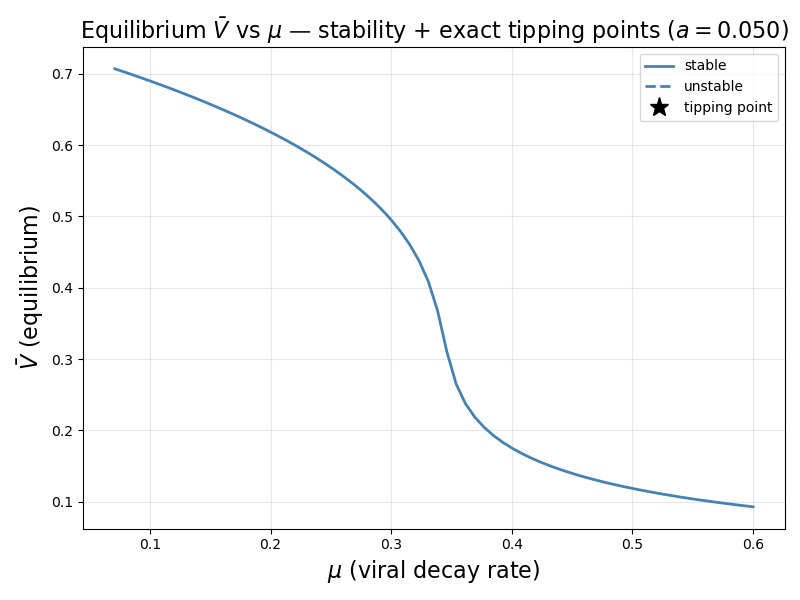}
        \caption{}
        \label{fig:cusp_c}
    \end{subfigure}
    \caption{Equilibrium branches $\bar V$ versus $\mu$ for representative source strengths below, near, and above the cusp value $a^*$. 
    (a) $a<a^*$, where two saddle-node bifurcations are present.
    (b) $a\approx a^*$, where the two saddle-node bifurcations merge at a cusp point.
    (c) $a>a^*$, where the branch has no turning points.}
    \label{fig:cusp_comparison}
    \end{figure}

    To determine the kinetic stability of the equilibrium branches, we linearize \eqref{TBV_ODE} around
    $(\bar T,\bar B,\bar V)$, which yields the following Jacobian matrix
    \begin{equation}
    J\big|_{\mathrm{eq}}=
    \begin{pmatrix} 
    f_T & 0 & -\dfrac{\Phi}{1+\bar V}\\[2mm]
    0 & -\bar B & -2\eta\bar V\bar B\\[2mm]
    -\dfrac{\Lambda}{1+\bar T} & \dfrac{\Lambda}{\bar B} & \dfrac{\Lambda-a}{\bar V}
    \end{pmatrix}, 
    \label{eq:Jacobian}
    \end{equation}
    where
    \begin{equation}
    \Phi:= \frac{\alpha}{1+\bar V}, \qquad \Lambda:= \frac{\beta\bar V^2\bar B}{1+\bar T} = \mu\bar V-a.
    \label{eq:PhiLambda}
    \end{equation}  
    The characteristic polynomial of \eqref{eq:Jacobian} is
    \begin{equation} \lambda^3+c_1\lambda^2+c_2\lambda+c_3=0, \label{eq:charpoly_ode} 
    \end{equation}
    where 
    \begin{subequations}
    \begin{align}
    c_1 &= -f_T+\bar B-\mu+\frac{2a}{\bar V}, \label{eq:c1}\\
    c_2 &= \left[2\eta\bar V\Lambda-\frac{\bar B(\Lambda-a)}{\bar V}\right]+\left[\frac{f_T(\Lambda-a)}{\bar V}-\frac{\Phi\Lambda}{(1+\bar T)(1+\bar V)}\right]-f_T\bar B, \label{eq:c2}\\
    c_3 &= -f_T\left[2\eta\bar V\Lambda-\frac{\bar B(\Lambda-a)}{\bar V}\right]-\frac{\Phi\bar B\Lambda}{(1+\bar T)(1+\bar V)}. \label{eq:c3}
    \end{align}
    \end{subequations}

    By the Routh--Hurwitz criterion, the homogeneous equilibrium is linearly stable under the reaction kinetics if
    \begin{equation}
    c_1>0, \qquad c_2>0, \qquad c_3>0, \qquad c_1c_2>c_3.
    \label{eq:routhhurwitz}
    \end{equation}
    
    Applying the Routh--Hurwitz conditions \eqref{eq:routhhurwitz} along the equilibrium branch shows that, for $a=0.02$, the solid branches in Figure~\ref{fig:bifurcation_a} are kinetically stable, while the middle branch is unstable. Consequently, for $\mu_1<\mu<\mu_2,$ the system exhibits bistability between two positive homogeneous equilibria.
    
    To verify these stability results numerically, we perturbed a few equilibria along the bifurcation branch using \textit{FlexPDE}. Perturbations around stable equilibria decayed back to the steady state, whereas those around unstable equilibria moved away and were attracted to the nearby stable state. The triangular markers in the left panel of Figure~ \ref{fig:bifurcation_a} indicate the equilibria used in these tests, which agree with stability predictions given in \eqref{eq:routhhurwitz}.

    
    \subsection{Diffusion-driven instability}
    \label{subsec:turing}
    We now examine the effect of diffusion on the stability of the homogeneous equilibria identified above. The analysis is carried out on the fixed domain $x\in[-1,1]$ subject to Neumann boundary conditions, $B_x(\pm1,t)=V_x(\pm1,t)=0.$

    Let $\bar{\mathbf u}=(\bar T,\bar B,\bar V)^{\mathrm T}$ be a kinetically stable homogeneous equilibrium. We introduce the perturbations
    \begin{equation}
    \mathbf u(x,t)=\bar{\mathbf u}+ e^{\lambda_m t}\cos(\nu x),\qquad \nu=(m\pi)^2, \qquad m=0,1,2,\ldots.
    \label{eq:perturbation}
    \end{equation}
    For the $m$-th spatial mode, the linearized Jacobian is
    \begin{equation}
    J_m=J-\nu D,\qquad D=\operatorname{diag}(0,D_{B_L},D_{V_L}), \qquad D_{V_L}:=\varepsilon_L^2
    \label{eq:Jm}
    \end{equation}
    where $J$ is the kinetic Jacobian in \eqref{eq:Jacobian}. The corresponding $\lambda_m$ satisfies
    \begin{equation}   \lambda_m^3+C_1(\nu)\lambda_m^2+C_2(\nu)\lambda_m+C_3(\nu)=0,
    \label{eq:charpoly_pde}
    \end{equation}
    where
    \begin{align}
    C_1(\nu)&=c_1+(D_{B_L}+D_{V_L})\nu,\label{eq:C1m}\\
    C_2(\nu)&=c_2+\left[D_{V_L}(\bar B-f_T)-D_{B_L}(f_T+q_0)\right]\nu+D_{B_L}D_{V_L}\nu^2,\label{eq:C2m}\\
    C_3(\nu)&=c_3+\left\{D_{B_L}\left[f_T\left(\mu-\frac{2a}{\bar V}\right)-\frac{\Phi\Lambda}{(1+\bar T)(1+\bar V)}\right]-D_{B_L}\bar Bf_T\right\}\nu-f_TD_{B_L}D_{V_L}\nu^2,\label{eq:C3m}
    \end{align}
    and $C_j(0)=c_j, j=1,2,3$, as given in \eqref{eq:c1}-\eqref{eq:c3}. For each spatial mode, the Routh--Hurwitz conditions are
    \begin{equation}
    C_1(\nu)>0,\qquad C_2(\nu)>0,\qquad C_3(\nu)>0, \qquad
    H(\nu):=C_1(\nu)C_2(\nu)-C_3(\nu)>0.\label{eq:RH_pde}
    \end{equation}
    Since the underlying homogeneous equilibrium is kinetically stable, these inequalities hold at $m=0$. Moreover, $C_1(\nu)=c_1+(D_{B_L}+D_{V_L})\nu>0.$  Hence diffusion-driven instability can arise either through the vanishing of \(C_3\), corresponding to a stationary Turing instability, or through the vanishing of \(H\) while \(C_2\) and \(C_3\) remain positive, corresponding to an oscillatory Turing instability.
    
    \subsubsection{Stationary Turing instability.}
    \label{sec: stationary}
    A stationary Turing instability is characterized by a real eigenvalue crossing through zero and hence by $C_3(\nu)=0.$ Writing
    \begin{equation} \label{eq:C3_quadratic}
    C_3(\nu)=c_3+p\nu+q\nu^2,
    \end{equation}
    where
    \begin{equation}\label{eq:pq}
    p=D_B\left[f_T\left(\mu-\frac{2a}{\bar V}\right)-\frac{\Phi\Lambda}{(1+\bar T)(1+\bar V)}\right]-D_V\bar Bf_T,\qquad
    q=-f_TD_BD_V,
    \end{equation}
    where $f_T, \Phi, \Lambda$ are given in \eqref{eq:fT} and \eqref{eq:PhiLambda}. With the parameter regime considered in \eqref{eq:fT_condition} we have $f_T<0$, and $q>0$. Since kinetic stability also gives $c_3>0$,  in the continuous dispersion relation, the onset of a stationary Turing instability occurs when the minimum of $C_3(\nu)$ first decrease to zero, which implies that $C_3(\nu_c)=0, C_3'(\nu_c)=0,$ or equivalently,
    \begin{equation}
    p^2-4qc_3=0,\qquad \nu_c=-\frac{p}{2q}>0.
    \label{eq:stationary_threshold}
    \end{equation}
    Beyond this threshold, an instability band with \(C_3(\nu)<0\) exists with
    \begin{equation}
    \nu_\pm=\frac{-p\pm\sqrt{p^2-4qc_3}}{2q}. \label{eq:nu_pm}
    \end{equation}
    Hence, a stationary spatial instability occurs whenever an admissible Neumann eigenvalue lies in $(\nu_-,\nu_+)$.

    Applying the stationary instability criterion \eqref{eq:stationary_threshold} along the kinetically stable equilibrium branch, in Fig. \ref{fig:bifurcation_b} we find a stationary Turing bifurcation on the upper stable branch at $\mu^*\approx0.12, \bar V^*\approx 0.64$. To verify this threshold, we also perform direct numerical simulations of the full system \eqref{TBV_main} using \textit{FlexPDE} for values of $\mu$ on either side of $\mu^*$. As shown in Figure~\ref{fig:turing_sim1}, the homogeneous state is stable for $\mu<\mu^*$, whereas a stationary spatial pattern develops for $\mu>\mu^*$, which is consistent with the predicted Turing threshold.
    
	\begin{figure}[h]
		\centering
		\includegraphics[width=\textwidth]{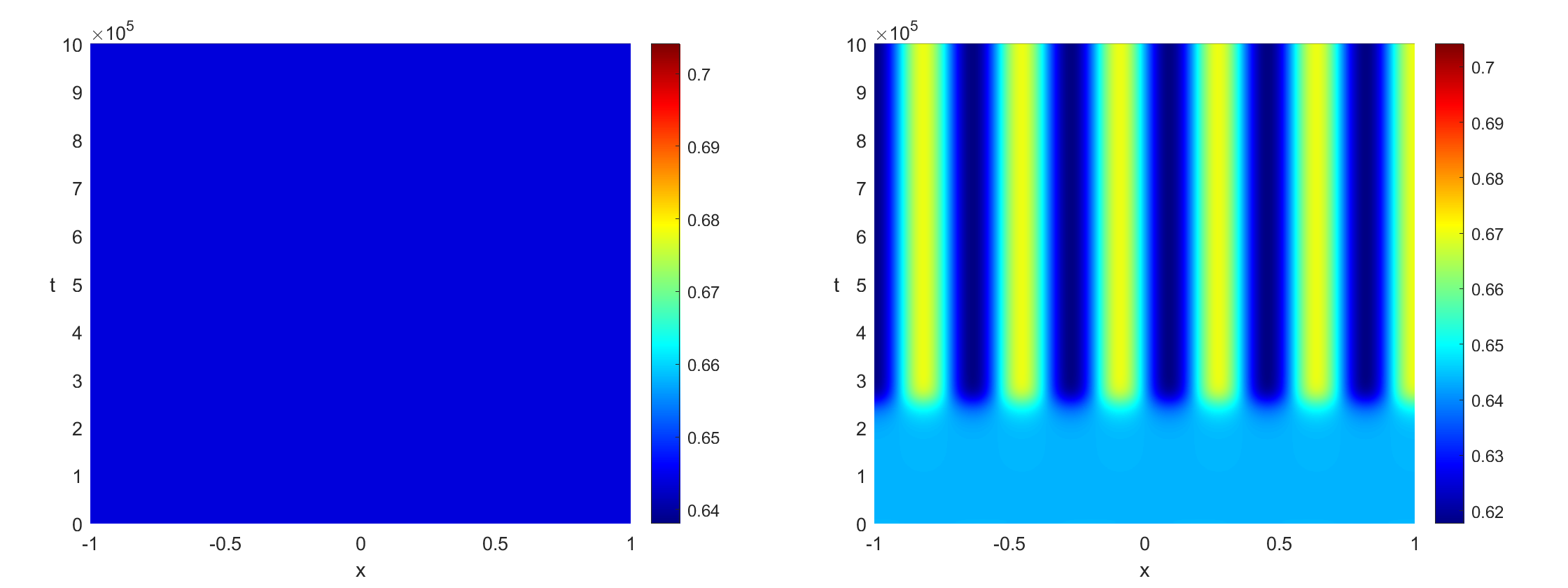}
        \caption{Space-time evolution of the viral concentration $V(x,t)$ on the fixed domain $x\in[-1,1]$ for values of $\mu$ on either side of the stationary Turing bifurcation point $\mu^*\approx0.1218$. Left: $\mu=0.121$, for which the perturbation decays and the solution returns to the homogeneous equilibrium. Right: $\mu=0.122$, for which a stationary spatial pattern develops. The remaining parameters are given in Table~\ref{tab:params}. The initial condition is $T=0.80$, $B=0.037$, and $V=0.70+0.01\cos(10x)$.}
        \label{fig:turing_sim1}
    \end{figure}
    Alternatively, fixing $\mu$ and using the diffusion ratio $r:=\frac{D_{B_L}}{D_{V_L}}$, the stationary instability condition \eqref{eq:stationary_threshold} can be expressed as a quadratic inequality in $r$, with instability occurs when
    \begin{equation}
    X^2r^2+2f_T\left(X\bar B+2c_3\right)r+(f_T\bar B)^2>0,\qquad X:=\frac{\Phi\Lambda}{(1+\bar T)(1+\bar V)}-f_T\left(\mu-\frac{2a}{\bar V}\right).
    \label{eq:DB_quadratic}
    \end{equation}
    
    This condition shows that a stationary Turing instability can arise once the diffusion ratio $r$ is sufficiently large. Since $\frac{D_{B_L}}{D_{V_L}}=\frac{D_B}{D_V},$ the corresponding threshold in $r$ is independent of the domain size $L$. Thus, domain growth does not shift the critical diffusion ratio itself, but instead affects the available spatial modes. As $L$ increases, these modes become more closely spaced, so that a mode may enter the unstable band even though the critical diffusion ratio itself remains fixed.
    
    For the parameter values in Table~\ref{tab:params}, with $a=0.02$, the relevant threshold in the regime is $D_B/D_V=r>r^*\approx29.31.$ Thus, for fixed $D_{V_L}$, the onset of the stationary Turing instability is equivalently described by the critical value $D_{B}^*=r^*D_{V}.$ This transition is illustrated in Figure~\ref{fig:DB_compare}, where $D_{B}$ is chosen slightly below and above $D_{B}^*\approx0.003811$. For $D_{B}<D_{B}^*$, a small perturbation of the homogeneous steady state decays and the solution returns to the spatially uniform state. In contrast, for $D_{B}>D_{B}^*$, the homogeneous state becomes unstable and a stationary Turing pattern develops. The numerical simulations are consistent with the threshold predicted by the linear stability analysis.
    
        \begin{figure}[h]
		\centering
		\includegraphics[width=\textwidth]{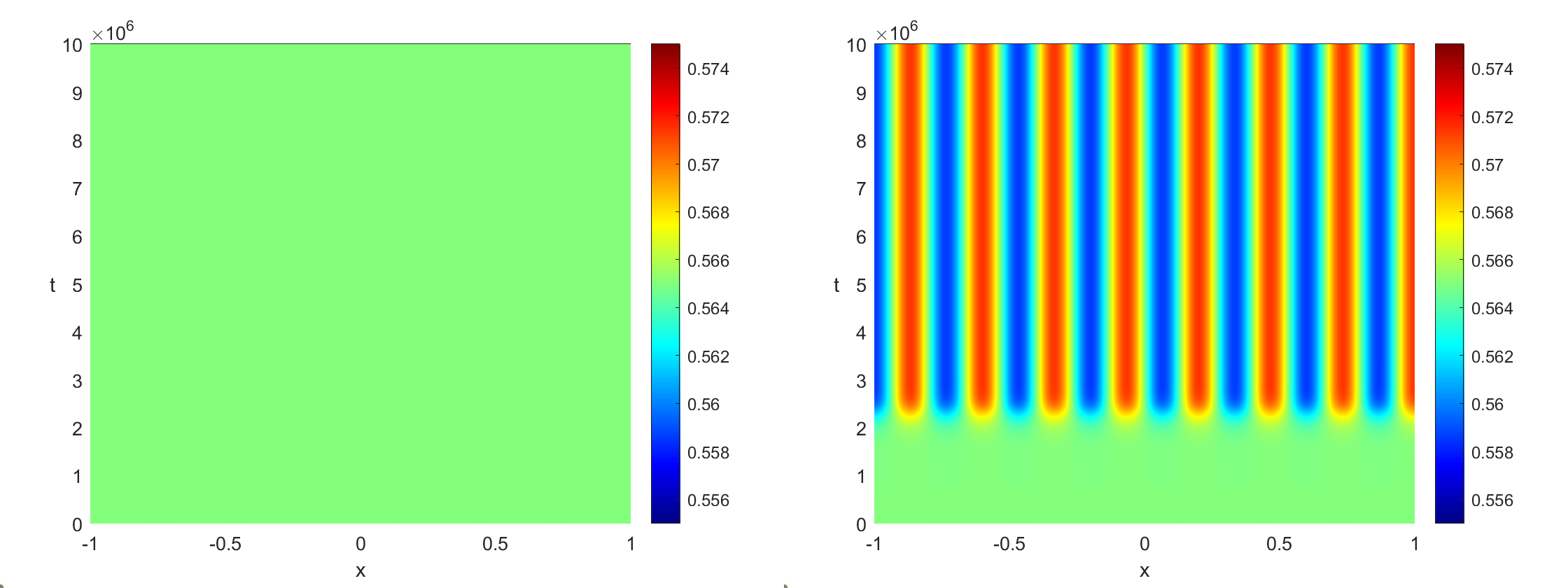}
		\caption{Space--time plots of $V(x,t)$ on $x\in[-1,1]$ for $D_B$ below and above the critical value $D_B^*\approx0.003811$. Left: $D_B=0.0038105<D_B^*$, no pattern forms. Right: $D_B=0.0038111>D_B^*$, a stationary stripe pattern emerges. The parameters are given in table \ref{tab:params}, with the initial condition $T=0.86$, $B=0.36$, $V=0.57 + 0.01\cos(10x)$.}
		\label{fig:DB_compare}
	\end{figure}
    \subsubsection{Oscillatory Turing instability.}\label{sec:oscillatory}
    We next consider whether diffusion can destabilize the homogeneous equilibrium through a complex conjugate pair of eigenvalues. For the cubic characteristic polynomial \eqref{eq:charpoly_pde}, an oscillatory crossing is characterized by
    \begin{equation}
    H(\nu):=C_1(\nu)C_2(\nu)-C_3(\nu)=0,
    \label{eq:H_condition}
    \end{equation}
    together with
    \begin{equation}
    C_2(\nu)>0,\qquad C_3(\nu)>0.
    \label{eq:osc_positive}
    \end{equation}
    At such a point, the critical eigenvalues are \(\lambda=\pm i\omega\), with $\omega^2=C_2(\nu)>0.$
    
    From \eqref{eq:H_condition}, \(H\) is a cubic polynomial in \(\nu\), which can be expressed as
    \begin{equation}
    H(\nu)=A_3\nu^3+A_2\nu^2+A_1\nu+A_0.
    \label{eq:Hpoly}
    \end{equation}
    Writing $C_2(\nu)=c_2+b\nu+d\nu^2,  C_3(\nu)=c_3+p\nu+q\nu^2,$     with $b=D_{V_L}(\bar B-f_T)-D_{B_L}\left(f_T+\mu-\frac{2a}{\bar V}\right),\, d=D_{B_L}D_{V_L},$    the coefficients are
    \begin{subequations}\label{eq:Hcoeff}
    \begin{align}
    A_0&=c_1c_2-c_3,\\
    A_1&=c_1b+(D_{B_L}+D_{V_L})c_2-p,\\
    A_2&=c_1d+(D_{B_L}+D_{V_L})b-q,\\
    A_3&=(D_{B_L}+D_{V_L})d.
    \end{align}
    \end{subequations}
    For every kinetically stable equilibrium, $A_0>0, A_3>0,$ and hence
    \[H(0)>0,\qquad H(\nu)\to+\infty \quad\text{as}\quad \nu\to\infty.\]
    Therefore, an oscillatory instability can occur only if \(H\) develops a local minimum at some \(\nu>0\) that reaches zero while \(C_2(\nu)\) and \(C_3(\nu)\) remain positive. 
    Under the small-virus-diffusion regime $D_{V_L}=\mathcal{O}(\varepsilon^2), 0<\varepsilon\ll1.$ Define \begin{equation}
    s:=f_T+\mu-\frac{2a}{\bar V},
    \label{eq:s_def}
    \end{equation}
    from \eqref{eq:c1}, $c_1=\bar B-s$, then expansion of \eqref{eq:Hcoeff} as $D_{V_L}\to0$ gives
    \begin{align}
    A_1&=D_{B_L}\left[s^2-2\bar Bs+2\eta\bar V\Lambda\right] +\mathcal O(D_{V_L}),\label{eq:A1_smallDV}\\
    A_2&=-D_{B_L}^2s+\mathcal O(D_{V_L}).
    \label{eq:A2_smallDV}
    \end{align}
    If \(s<0\), then $s^2-2\bar Bs+2\eta\bar V\Lambda>0,$ and hence, for sufficiently small \(D_{V_L}\), $A_1>0, A_2>0.$
    Together with \(A_0>0\) and \(A_3>0\), this implies
    \begin{equation}
    H(\nu)>0,\qquad
    \nu\geq0.
    \label{eq:noosc}
    \end{equation}
    Thus, an oscillatory Turing instability cannot occur on this portion of the kinetically stable equilibrium branch.

    For the remaining parameter range, where $s\geq0$, the preceding asymptotic argument does not by itself exclude an oscillatory instability. We therefore evaluate the full conditions \eqref{eq:H_condition}--\eqref{eq:osc_positive} numerically along the kinetically stable portions of the equilibrium branch shown in Figure~\ref{fig:bifurcation_a}. At each equilibrium, we compute the nonnegative real roots of $H(\nu)$ and evaluate $C_2(\nu)$ and $C_3(\nu)$ at these roots. No root is found for which $C_2(\nu)>0, C_3(\nu)>0.$ Hence, no oscillatory diffusion-driven instability is detected along the remaining portion of the kinetically stable equilibrium branch.

    Combining this numerical verification with the analytical exclusion for $s<0$ in the regime $D_{V_L}\ll1$, we find no independent oscillatory Turing instability along the kinetically stable equilibrium branch considered here. Therefore, in the rest of the paper, we focus on the stationary Turing instability and the nonlinear localized-pattern regime.

	\section{The Inner Problem: Onset of Spike Self-Replication} \label{sec:inner}
    
    Having established the stationary Turing instability of the homogeneous state, we now turn to the nonlinear localized-pattern regime, in which the viral concentration develops localized spikes. As the domain grows, new spikes may arise through two distinct mechanisms: an existing spike may split into two through self-replication, or a new spike may form from the background through nucleation. In this section, we explore self-replication, which is governed by the local structure and stability of the spike core. Following the asymptotic framework of \cite{gai2025asymptotic}, we derive the corresponding inner problem for the TBV model and determine the local threshold for spike self-replication.
    
    We begin with the system \eqref{eq:T-main}--\eqref{eq:V-main}. Since the $T$-component is non-diffusive, its quasi-steady value is determined pointwise by $V$ through $T=\Theta(V),$ defined in \eqref{eqn:theta(v)}. The system then reduces to
    \begin{subequations}\label{eq:BV_reduced}
    \begin{align}
    B_t&= D_{B_L}B_{xx} +(1-B)B-\eta V^2B, \qquad -1<x<1, \label{eq:B-reduced}\\
    V_t &= \varepsilon_L^2V_{xx} +\frac{\beta}{1+\Theta(V)}V^2B -\mu V+a, \qquad -1<x<1, \label{eq:V-reduced} 
    \end{align} 
    \end{subequations}
    subject to $B_x(\pm1,t)=V_x(\pm1,t)=0,$    where
    \[ D_{B_L}=\frac{D_B}{L^2}, \qquad \varepsilon_L=\frac{\varepsilon}{L}, \qquad \varepsilon^2=D_V\ll1. \]
    
    The behavior of $\Theta(V)$ for large $V$ will be needed in the spike core, so we first establish the following property.   
    \begin{proposition}\label{prop:theta-limit}
    Define 
    \begin{equation} 
    \phi(T):=\delta T-\frac{\rho T^2}{k^2+T^2}.
    \label{eq:phi_T}
    \end{equation}
    Suppose that \eqref{eq:fT_condition} holds, so that $f_T<0$ for $T\geq0$. Then, for each $V\geq0$, the equation $\phi(T)=\frac{\alpha}{1+V}$ admits a unique positive solution $T=\Theta(V)$. Moreover,
    \begin{equation}
    \Theta(V)\to0 \qquad\text{as}\qquad V\to\infty. \label{eq:Theta_limit}
    \end{equation} 
    \end{proposition}
    
    \begin{proof}
    Since 
    \[\phi'(T)=\delta-\frac{2\rho k^2T}{(k^2+T^2)^2}=
    -f_T>0,\]
    the function $\phi$ is strictly increasing for $T\geq0$. Moreover,
    \[\phi(0)=0, \qquad \phi(T)\to\infty \quad\text{as}\quad T\to\infty. \]
    Hence, for every $V\geq0$, the equation $\phi(T)=\frac{\alpha}{1+V}>0$
    has a unique positive solution, denoted by $T=\Theta(V)$.
    
    As $V\to\infty$, $\frac{\alpha}{1+V}\to0,$ and  $\phi$ is continuous and strictly increasing with $\phi(0)=0$, its inverse is continuous on its range. Therefore,
    \[\Theta(V)=\phi^{-1}\left(\frac{\alpha}{1+V}\right) \to \phi^{-1}(0) =0.\]
    \end{proof}

    To seek quasi-steady localized solutions of \eqref{eq:BV_reduced}, we first construct a single spike centered at $x=0$ on $(-\ell,\ell)$, where $\ell$ denotes the half-width of the spike cell. For an equally spaced $K$-spike configuration on $(-1,1)$, $\ell=1/K$, and the full pattern follows by reflection and periodic extension. The onset of self-replication is determined locally by the structure of an individual spike and can therefore be studied from the corresponding single-spike core problem.
    
    Near a spike centered at $x=0$, we introduce the inner scaling
    \begin{equation}
    x=\frac{\varepsilon_L}{\sqrt{\mu}}\,y,
    \qquad
    V=\frac{\sqrt{\mu D_{B_L}}}{\varepsilon_L\sqrt{\eta}}\,\mathcal{V},
    \qquad
    B=\frac{\varepsilon_L\sqrt{\mu\eta}}
    {\beta\sqrt{D_{B_L}}}\,\mathcal{B}.
    \label{eq:inner_scaling}
    \end{equation}
    
    
    In this region, $V$ becomes asymptotically large as $\varepsilon_L\to0$, and hence Proposition~\ref{prop:theta-limit} yields $\frac{1}{1+\Theta(V)}\to1.$ Therefore, to leading order, the quasi-steady problem \eqref{eq:B-reduced}--\eqref{eq:V-reduced} reduces to the coupled core problem
    \begin{equation}
    0= \mathcal{B}_{yy}- \mathcal{V}^2\mathcal{B}, \qquad
    0 = \mathcal{V}_{yy}  -\mathcal{V} + \mathcal{V}^2\mathcal{B}, \qquad y\ge 0.
    \label{eq:core-canonical}
    \end{equation}
    We seek an even single-spike solution satisfying
    \begin{equation}
    \mathcal{V}_y(0) =\mathcal{B}_y(0)=0,\qquad \mathcal{V}(y)\to0 \quad\text{as }y\to\infty.
    \label{eq:core_bc}
    \end{equation}
    Integrating the second equation in \eqref{eq:core-canonical} gives the far-field behavior
    \begin{equation}
    \mathcal{B}(y) \sim \mathcal A y+C_s+o(1), \qquad y\to\infty, \qquad \mathcal A := \lim_{y\to\infty}\mathcal{B}_y=\int_0^\infty \mathcal{V}^2\mathcal{B}\,dy.
    \label{eq:B_farfield}
    \end{equation}
    The slope $\mathcal A>0$ provides a natural parameterization of the family of core solutions, while the constant $C_s=C_s(\mathcal A)$ is determined numerically.
    
    
    \subsection{Core bifurcation and the self-replication threshold} \label{sec:core-numerics}
    To parametrize the core solution branch, we define
    \begin{equation}
    \gamma:=\mathcal{B}(0)\mathcal{V}(0),
    \label{eq:gamma_def}
    \end{equation}
    and then solve \eqref{eq:core-canonical} numerically on a truncated interval $0\leq y\leq Y_{\max}$, subject to
    \[ \mathcal{V}_y(0)=\mathcal{B}_y(0)=0, \qquad \mathcal{V}(Y_{\max})=0, \qquad \mathcal{B}_y(Y_{\max})=\mathcal A, \]
    and numerically continue the resulting family of solutions. Figure~\ref{fig:core_bif} shows the bifurcation diagram of $\gamma$ versus the far-field slope $\mathcal A$. Along the upper solution branch, the single-spike solution ceases to exist at a saddle-node bifurcation as $\mathcal A$ is increased through $\mathcal A=\mathcal A^*\approx1.35.$ Along the lower branch, the solution is volcano-shaped. As discussed below, the saddle-node value $\mathcal A^*$ marks the onset of spike self-replication, which is consistent with the corresponding thresholds obtained for the Gray-Scott \cite{kolokolnikov2005existence, muratov2000static}.
    
    \begin{figure}[h]
    \centering
    \includegraphics [width=\textwidth]{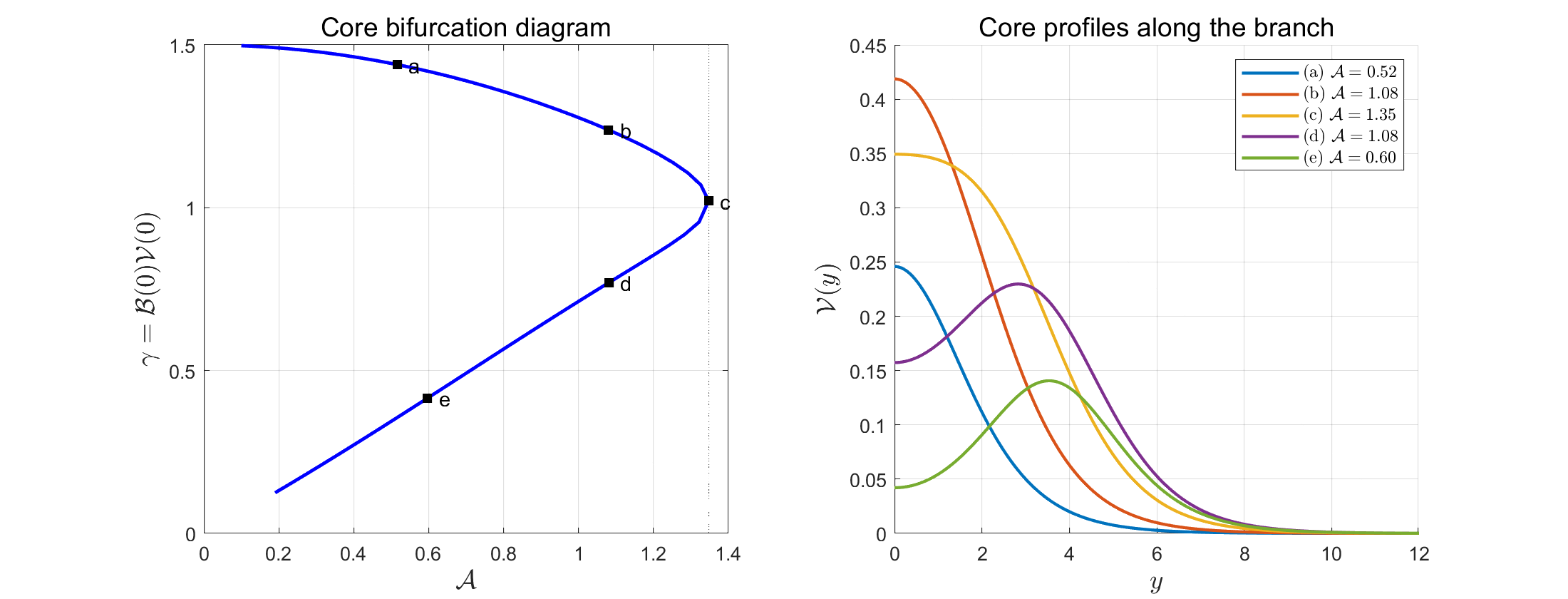} 
    \caption{Left: Bifurcation diagram of single-spike solutions of the
    canonical core problem, showing
    $\gamma=\mathcal B(0)\mathcal V(0)$ as a function of the far-field slope $\mathcal A$. The branch undergoes a saddle-node bifurcation at $\mathcal A=\mathcal A^*\approx1.35$.
    Right: Representative profiles of the viral core component
    $\mathcal V(z)$ at selected points along the solution branch, illustrating the transition from a single-peaked spike to a volcano-shaped profile beyond the fold.} \label{fig:core_bif} \end{figure}

   The stability of the core problem has been studied previously in \cite{doelman2006homoclinic,kolokolnikov2005existence,kolokolnikov2005pulse,nishiura1999skeleton}.  The upper solution branch is linearly stable up to the saddle-node, whereas the lower branch is unstable. At the fold, the dominant eigenvalue passes through zero, and the corresponding neutral eigenfunction has a characteristic dimple-shaped profile. Therefore,  $\mathcal A=\mathcal A^*\approx1.35$ provides the local core threshold for the onset of spike self-replication.

   Along the upper solution branch in Fig.~\ref{fig:core_bif}, we compute the far-field constant $C_s$ from \eqref{eq:B_farfield}, together with the core value $\mathcal B(0)$, as functions of the far-field slope $\mathcal A$. The resulting curves are shown in Fig.~\ref{fig:Cs_core}. In particular, $C_s$ remains positive along the entire upper branch. These quantities will enter the outer matching condition derived in the next section.

   \begin{figure}[htbp]
    \centering
    \includegraphics[width=0.7\textwidth]{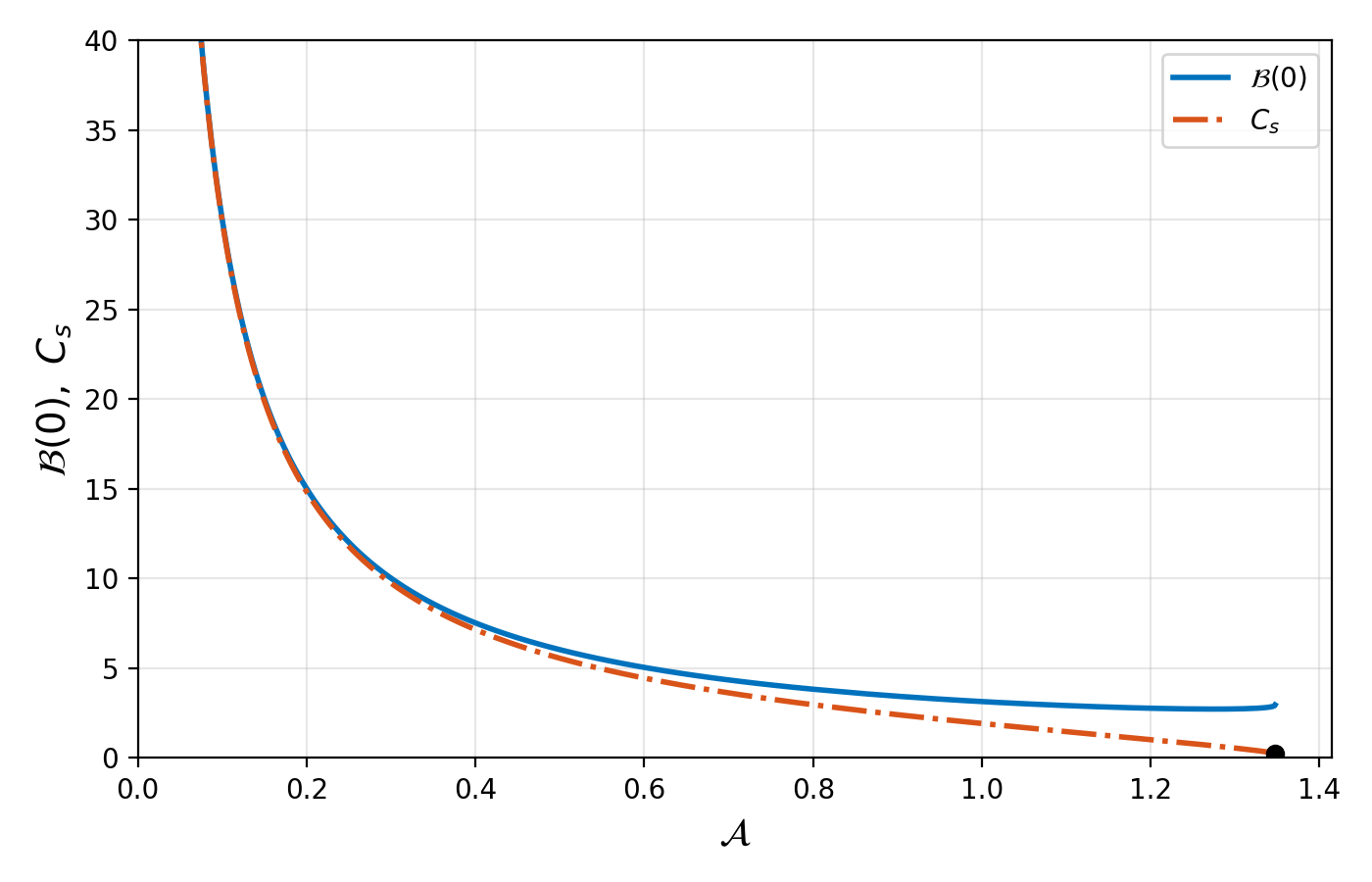} 
    \caption{Plots of $C_s$ and $\mathcal{B}(0)$ calculated from \eqref{eq:core-canonical} and \eqref{eq:B_farfield}. Note that $C_s>0$ up to the fold point $\mathcal{A}_c=1.347$, where $C_s\approx0.246$.}
    \label{fig:Cs_core} 
    \end{figure}

   To determine when this threshold is reached on a growing domain, the core solution must be matched to the outer solution through \eqref{eq:B_farfield}. This is carried out in the next section, where the resulting matching condition relates $\mathcal A$ to the domain length $L$ and determines the critical replication length $L_{\mathrm{rep}}$.
    
	 \section{The Outer Problem: Onset of Spike Nucleation} \label{sec:outer} 
     We now turn to the outer region away from the spike core. Since the spike solution is even about $x=0$, it suffices to consider the interval $O(\varepsilon_L)\ll x<\ell.$ We will show that, as the domain increases, the corresponding outer problem loses solvability that leads to spike nucleation, as shown in \cite{gai2025asymptotic, gai2024nucleation}. Thus, whether new spikes are generated by self-replication or nucleation is determined by which of the inner and outer thresholds is reached first as domain increases.
    
    Away from the spike core, the quasi-steady outer problem satisfies
    \begin{align}
    0 &= \frac{\beta}{1+\Theta(V)}V^2B-\mu V+a,\quad 0^+<x<l,\label{eq:outer-V}\\
    0 &= D_{B_L}B_{xx}+(1-B)B-\eta V^2B, \qquad B_x(\ell)=0. \label{eq:outer-B}
    \end{align}
    
    For the outer analysis, where $V=\mathcal{O}(1)$ and $T=\mathcal{O}(1)$, the nonlinear term $\frac{\rho T^2}{k^2+T^2}$ is small in the regime considered here, where $k$ is relatively large and $\rho$ is small, which is the regime considered in the paper. We therefore approximate
    \begin{equation}
    \Theta(V)\sim h(V):=\frac{\alpha}{\delta(1+V)}.
    \label{eq:h-def}
    \end{equation}
    Hence, we have from \eqref{eq:outer-V} that
    \begin{equation}
    B=G(V) := \frac{(\mu V-a)(1+h(V))}{\beta V^2},\qquad V>V_*:=\frac{a}{\mu}. \label{eq:G-def}
    \end{equation}
    Defining $g(V):=G'(V)$, we obtain
    \begin{equation}
   g(V) = \frac{(1+h(V))(2a-\mu V) +(\mu V-a)Vh'(V)}{ \beta V^3}. 
    \label{eq:g-identity}
    \end{equation}
    Substituting $B=G(V)$ into \eqref{eq:outer-B} then yields the scalar outer problem
    \begin{equation}
    D_{B_L}\big(g(V)V_x\big)_x=R(V), \qquad 0^+<x<\ell,
    \qquad
    V_x(\ell)=0,
    \label{eq:outer-ODE}
    \end{equation}
    where
    \begin{equation}
    R(V) := G(V)\big(\eta V^2-1+G(V)\big).
    \label{eq:R-def}
    \end{equation}
	
	As $x\to 0^+$, the outer solution matches the far-field behavior of the core solution, so using \eqref{eq:B_farfield} and the inner scaling, we obtain
    \begin{align} 
    g(V)V_x&=B_x \to\frac{\mu\sqrt{\eta}}{\beta\sqrt{D_{B_L}}}\,\mathcal A, \quad \text{and} \label{eq:match-slope}\\
    G(V)&=B \to \frac{\sqrt{\mu\eta}}{\beta} \frac{\varepsilon}{\sqrt{D_B}}\, C_s(\mathcal A),\quad \text{as} \quad x\to 0^+.
    \label{eq:match-value}
    \end{align}
    Since $G(V_*)=0$, with $V_*=a/\mu$, expansion of \eqref{eq:match-value} about $V=V_*$ gives \begin{equation}
    V(0^+) = V_* + \frac{V_*^2}{1+h(V_*)} \frac{\varepsilon_L}{\sqrt{D_B}} \sqrt{\frac{\eta}{\mu}}\, C_s(\mathcal A) +\mathcal{O}(\varepsilon_L^2). \label{eq:V0plus} 
    \end{equation}
    We now discuss the solvability of the outer problem \eqref{eq:outer-ODE} subject to the matching conditions \eqref{eq:match-slope} and \eqref{eq:V0plus}.  We first identify the available range of the outer solution and then establish its monotonicity.
    \begin{lemma}\label{lem:g}
    Given $h(V)=\frac{\alpha}{\delta(1+V)}, V_*=\frac{a}{\mu},$ and suppose
    \begin{equation}
    \frac{a}{\mu}<1+\frac{\alpha}{\delta}.
    \label{eq:g-condition}
    \end{equation}
    Then $G(V)>0$ for $V>V_*$, and $g(V)=G'(V)$ has a unique positive zero $V_g$ satisfying $V_*<V_g<\frac{2a}{\mu}.$  Moreover, $g(V)>0,$ for $V_*\leq V<V_g.$
    \end{lemma}
    
    \begin{proof}
    Since $G(V)= \frac{(\mu V-a)(1+h(V))}{\beta V^2},$ we have $G(V)>0$ for $V>V_*$. From \eqref{eq:g-identity}, we have $g(V_*)>0, $ and $g\left(\frac{2a}{\mu}\right)<0,$ due to the fact that $h'(V)<0$. Hence $g$ has at least one zero in $(V_*,2a/\mu)$.
    
    To establish uniqueness, solving $g(V)=0$ yields
    \[P(V)=-\mu V^3-2\left(\mu(1+\frac{\alpha}{\delta})-a\right)V^2+\left(a(4+3\frac{\alpha}{\delta})-\mu(1+\frac{\alpha}{\delta})\right)V+2a(1+\frac{\alpha}{\delta})=0.\]
    Under \eqref{eq:g-condition}, the coefficients of $P$ have signs $(-,-,\pm,+),$ and hence exactly one sign change. By Descartes' rule of signs, $P(V)$, and therefore $g(V)$, has exactly one positive zero. Consequently,
    \[V_*<V_g<\frac{2a}{\mu},\] 
    and $g(V)>0$ for $V_*\leq V<V_g$.
    \end{proof}
    
    We next determine the range on which the outer solution is monotone. Using \eqref{eq:R-def}, we write $R(V)=G(V)f(V),$ where $ f(V):=\eta V^2-1+G(V).$ Since $G(V)>0$ for $V>V_*$, the sign of $R(V)$ on this interval is determined by $f(V)$. The following lemma identifies the first positive zero of $f$ above $V_*$.    
    \begin{lemma}\label{lem:f} 
    The equation $f(V)=0$ has either one or three positive roots. Let $V_f$ denote the first root above $V_*=\frac{a}{\mu}.$ Suppose that 
    \begin{equation}
    \eta<\left(\frac{\mu}{a}\right)^2.
    \label{eq:f-condition}
    \end{equation}    
    Then $f(V)<0$ for $V_*<V<V_f.$ Consequently, $R(V)<0, V_*<V<V_f$.
    \end{lemma}

    \begin{proof}
    Multiplying $f(V)=0$ by $\beta\delta V^2(1+V)$ gives 
    \[0=\eta\beta\delta V^5+\eta\beta\delta V^4-\beta\delta V^3 +\delta(\mu-\beta)V^2 +\big[\mu(\alpha+\delta)-a\delta\big]V -a(\alpha+\delta).\]
    The coefficient signs are $(+,+,-,\pm,\pm,-)$, and hence there are either one or three sign changes. By Descartes' rule of signs, $f(V)=0$ therefore has either one or three positive roots, with multiplicities counted.
    
    Since $G(V_*)=0$, and $f(V_*) =\eta V_*^2-1 =\eta\left(\frac{a}{\mu}\right)^2-1<0$ by \eqref{eq:f-condition}. Moreover, $f\left(\frac{1}{\sqrt{\eta}}\right) = G\left(\frac{1}{\sqrt{\eta}}\right)>0,$ since $1/\sqrt{\eta}>V_*$ and $G(V)>0$ for $V>V_*$ by Lemma~\ref{lem:g}. Hence, by continuity, $f$ has at least one zero in $\left(V_*,\frac{1}{\sqrt{\eta}}\right).$
    
    Let $V_f$ denote the zero in $\left(V_*,\frac{1}{\sqrt{\eta}}\right)$. Then $f(V)<0,$ for $V_*<V<V_f.$ It follows that $R(V)=G(V)f(V)<0, V_*<V<V_f.$
    \end{proof}

    Although $V_f$ does not in general have a closed-form expression, it satisfies $V=\sqrt{(1-G(V))/\eta}$, so it lies close to $V_0:=\eta^{-1/2}$ whenever $G(V_0)$ is small. Expanding about $V_0$ (one Newton step) gives
    \begin{equation}
    V_f\approx V_0\left(1-\frac{G(V_0)}{2+V_0G'(V_0)}\right)=\frac{1}{\sqrt{\eta}}
    \left( 1-\frac{1}{2} G\left(\frac{1}{\sqrt{\eta}}\right) \right),
    \label{eq:Vf-approx}
    \end{equation}
    with error $\mathcal{O}\bigl(V_0\,G(V_0)^2\bigr)$. 
    
    We now establish the monotonicity of the outer solution.
    \begin{lemma}\label{lem:monotone}
    Suppose that $V_g<V_f$. Then any solution of \eqref{eq:outer-ODE} satisfying $V_x(\ell)=0$ and taking values in $(V_*,V_h)$ satisfies $V_x>0, $ on $0^+<x<\ell.$ Consequently,
    \[ V_\ell:=V(\ell) = \max_{0\leq x\leq\ell}V(x), \qquad V_\ell\leq V_g. \]
    \end{lemma}
    
    \begin{proof}
    Since $V_g<V_f$, Lemmas~\ref{lem:g} and \ref{lem:f} imply that $g\left(V\right)>0, f(V)<0,$ for $V_*<V<V_g.$ Then integrating \eqref{eq:outer-ODE} from $x$ to $\ell$ and using $V_x(\ell)=0$ gives
    \[ -D_{B_L}g\left(V(x)\right)V_x(x) = \int_x^\ell f(V(s))\,ds<0.\] 
    Therefore, $V_x(x)>0$ for $0<x<\ell$. Thus $V$ attains its maximum at $x=\ell$, yielding $V_\ell=V(\ell)\leq V_g$.
    \end{proof}

    Next, we multiply \eqref{eq:outer-ODE} by $g(V)V_x$ and integrate. Upon using the monotonicity of $V(x)$ on $ V_*< V < V_g$ together with $V_x(l) = 0$, we obtain
    \begin{equation}
		\frac{D_{B_L}}{2}\left(g(V)V_x\right)^2=\int_{V_l}^{V(x)}R(\xi)g(\xi)d\xi 
		\label{eq:first-integral}
	\end{equation}
    Setting $x \to 0^+$ in \eqref{eq:first-integral}, and using \eqref{eq:match-slope}, we conclude that
    	\begin{equation}
		\mathcal A^2=-\frac{2\beta^2}{\mu^2\eta}\int_{V(0^+)}^{V_l}R(\xi)g(\xi)d\xi,
		\label{eq:alg-A}
	\end{equation}
    where $V(0^+) \sim V_*=\frac{a}{\mu}$, but with $V(0^+) > V_*$ given in \eqref{eq:V0plus}. Since $R < 0$ when $V_* < V_l < V_g$, we observe that the right-hand side of \eqref{eq:alg-A} is positive on this range.

    To proceed solving the outer problem \eqref{eq:outer-ODE}, we define $\mathcal{G}'(\xi)$ by 
    \begin{equation}\label{eq:G_prime}
        \mathcal G'(\xi)\equiv-R(\xi)g(\xi)=-G(\xi)\big(\eta \xi^2-1+G(\xi)\big)G'(\xi)>0, \quad V_*<\xi<V_g.
    \end{equation}
    A first integral of \eqref{eq:G_prime} gives 
    \begin{equation}\label{eq:G-explicit}
    \mathcal G(\xi)= -\frac{1}{2}(\eta\xi^2-1)G(\xi)^2 -\frac{1}{3}G(\xi)^3 +\eta J(\xi),
    \end{equation}
    where
    \begin{align}
    J(\xi)\equiv\int\xi  G^2(\xi)d\xi=\frac{1}{\beta^2}\Bigg[
    &-\frac{a^2(1+c)^2}{2\xi^2} +\frac{2a(1+c)\bigl(ac+\mu(1+c)\bigr)}{\xi} +\frac{c^2(a+\mu)^2}{1+\xi} \nonumber\\
    &+\left(\mu^2(1+c)^2 +4a\mu c(1+c) +a^2c(3c+2) \right)\log\xi \nonumber\\ &-c(a+\mu)\bigl[a(3c+2)+\mu(c+2)\bigr] \log(1+\xi) \Bigg], \quad c\equiv\alpha/\delta.
    \end{align}
    
    In this way, \eqref{eq:first-integral} reduces to the explicit relation
    \begin{equation}
		\mathcal A^2=\frac{2\beta^2}{\mu^2\eta}\big(\mathcal G(V_\ell)-\mathcal G\left(V(0^+)\right)\big),
		\label{eq:first-integral_final}
	\end{equation}
    with $G(V(0^+))$ given in \eqref{eq:V0plus}. Since $\mathcal G^{\prime}(V_l) > 0$ on $V_* < V_l < V_g$, it follows that $\mathcal{A}$ is an increasing function of $V_l$.
	By Lemma \ref{lem:monotone}, on the range of $x$ where $V_* < V_l < V_g$, we use \eqref{eq:G_prime} in \eqref{eq:first-integral} to obtain that
    \begin{equation} \label{eq:vx}
    g(V) \, V_x = \sqrt{\frac{2}{D_{B_L}}} \, \sqrt{\mathcal G(V_l) - \mathcal G(V(x))} \, .
    \end{equation}
    By integrating this separable ODE, we get an implicit relation for $V(x)$ on $0^+ < x < \ell$ given by
    \begin{equation}\label{eq:sepa}
    \chi\left(V(x)\right) = \sqrt{\frac{2}{D_{B_L}}} \, x \, , \qquad \mbox{where} \qquad \chi\left(V(x)\right) \equiv \int_{V\left(0^+\right)}^{V(x)} \frac{g(\xi)}{\sqrt{\mathcal G(V_l) - \mathcal G(\xi)}} \, d\xi \,.
    \end{equation}
    Then, by setting $x = \ell$ and  $V(\ell) = V_l$ in \eqref{eq:sepa}, we obtain an implicit equation for $V_l$ given by
    \begin{equation}\label{eq:chi2}
    \chi(V_l) = \int_{V\left(0^+\right)}^{V_l} \frac{g(\xi)}
    {\sqrt{\mathcal G(V_l) - \mathcal G(\xi)}} \, d\xi =
    \sqrt{\frac{2}{D_{B_L}}} \ell \, ,
    \end{equation}
    where $V(0^+)$, with $V(0^+) > V_*$ given by \eqref{eq:V0plus}.
    
    Since the integral in \eqref{eq:sepa} is improper at $\xi = V_l$, to obtain a more tractable formula for $\chi(V_l)$ we integrate the expression in \eqref{eq:chi2} by parts by using $g(\xi) = -\mathcal G^{\prime}(\xi)/R(\xi)$. This yields the proper integral
    \begin{equation}\label{eq:properchi}
    \chi(V_l) = - 2 \frac{\sqrt{\mathcal G(V_l) - \mathcal G(V(0^+))}}{R(V(0^+))} + 2 \int_{V(0^+)}^{V_l}\, \frac{\sqrt{\mathcal G(V_l) - \mathcal G(\xi)}}{\left[R(\xi)\right]^2} \, R^{\prime}(\xi) \, d\xi \, ,
    \end{equation}
    where $R(\xi) = G(\xi)(\eta \xi^2-1+G(\xi))$ and $\mathcal G(\xi)$ is defined in \eqref{eq:G-def}. 
    \begin{lemma}\label{lem:chi}
    Suppose that $V_*<V_\ell<V_g.$ Then $\chi(V_\ell)$ is strictly increasing in $V_\ell$.
    \end{lemma}
    
    \begin{proof}
    Since $g(V)=G'(V)>0$ on $V_*<V_\ell<V_g$ by  lemma \ref{lem:g}, we may use $B=G(V)$ as a variable, and define
    \[Q(B):=-f(V(B))=1-\eta V(B)^2-B.\]
    On the interval considered, $Q(B)>0$, while $\frac{dQ}{dB}=-\frac{2\eta V+g(V)}{g(V)}<0.$ Moreover, since $\frac{d\mathcal G}{dB}=BQ(B),$ we can write
    \[ \chi(B_\ell) = \int_{B_0}^{B_\ell} \frac{dB} {\sqrt{\displaystyle\int_B^{B_\ell}s Q(s)\,ds}}, \quad
    B_0\equiv G(V_0^+),\quad B_l\equiv G(V_l).\]
    Scaling $B=B_\ell t$ yields
    \[ \chi(B_\ell) = \int_{B_0/B_\ell}^{1} \frac{dt} {\sqrt{\displaystyle\int_t^1 s Q(B_\ell s)\,ds}}. \]
    As $B_\ell$ increases, the lower limit decreases and, since $Q'(B)<0$, the integrand increases. Hence $d\chi/dB_\ell>0$. Finally, $dB_\ell/dV_\ell=g(V_\ell)>0$, and therefore $\chi'(V_\ell)>0.$
    \end{proof}
    
    From lemma \ref{lem:chi},  we conclude that $\chi(V_l)$ is a monotone increasing function of $V_l$ on $V(0^+) < V_l \leq V_g$ that reaches its maximum value at $V_l = V_g$.  As such, we define
    \begin{equation}\label{eq: chi_max}
    \chi_{\max} \equiv \chi(V_g) \, , \quad V_{\ell,\max} \equiv V_g \, .
    \end{equation}
    We summarize our asymptotic construction of a K-spike solution of radius $\ell$ to \eqref{TBV_main} as follows:
    \begin{proposition}\label{nuc_rep_prop}
    On the range of parameters given in \eqref{eq:g-condition} and \eqref{eq:f-condition}, so that $\frac{a}{\mu}<V_g < V_f$, where $V_g$ is the unique positive root of $g(V)=0$, and $V_f$ is the first positive root of $f(V)=0$. Our asymptotic construction of a one-spike solution to \eqref{TBV_main} on $|x| \leq \ell$ when $\varepsilon_L \ll 1$ reduces to solving the coupled nonlinear algebraic system
    \begin{equation}\label{eq:nas}
    \chi(V_l;\mathcal{A}) = \sqrt{\frac{2}{D_{B_L}}} \ell \, , \qquad
    \mathcal A^2=\frac{2\beta^2}{\mu^2\eta}\big(\mathcal G(V_\ell)-\mathcal G\left(V(0^+)\right)\big)\,,
    \end{equation}
    for $\mathcal{A}$ and $V_l = V(\ell)$ in terms of the parameters $\alpha, \beta, \delta, \mu, \eta, a,$ and $\ell/\sqrt{D_{B_L}}$. Here, $V(0^+)$ is given in \eqref{eq:V0plus}, while $\chi(V_l)$ and $\mathcal G(\xi)$ are defined in \eqref{eq:properchi} and \eqref{eq:G-explicit}, respectively. For a $K$-spike solution on $- 1 \leq x \leq 1$, we identify that $\ell = 1/K$.
    \end{proposition}
    
    
    \subsection{Absence of Spike Nucleation}
    
    We now consider the parameter regime such that $V_g>V_f$, where $V_g$ is the unique root of $g(V)=0$ defined in \eqref{eq:g-identity}, and $V_f$ is given in \eqref{eq:Vf-approx}. In this regime, the outer solution remains bounded away from $V_g$, and spike nucleation does not occur as the domain length increases. We first establish the following properties of the outer solution.
    
    \begin{lemma}\label{lemma:nonuc}
    Suppose that $V_g>V_f>V_*$, with $g(V)>0$ for
    $0<V<V_g$, and
    \[ R(V)<0 \quad\text{for }0<V<V_f, \qquad R(V)>0 \quad\text{for }V_f<V\leq V_g. \]
    Moreover, it follows from Lemma \ref{lem:g} that $g'(V_g)<0$. Then any positive, regular solution of
    \begin{equation}\label{eq:outer_nonuc}
    D_{B_L}\big(g(V)V_x\big)_x=R(V),
    \qquad 0<x<\ell, \qquad V(0)=V_*, \qquad V_x(\ell)=0, 
    \end{equation}
    satisfies $V_*<V(x)<V_f<V_g,$ and $V_x(x)>0,$ on $0<x<\ell.$
    \end{lemma}
    
    \begin{proof}
    We first show that $V(x)<V_g$ throughout the interval. If $V$ reaches $V_g$ at some point $x_0$, then evaluating \eqref{eq:outer_nonuc} at $x_0$ gives
    \[ D_{B_L}g'(V_g)V_x(x_0)^2=R(V_g). \]
    This is impossible since $g'(V_g)<0$ and $R(V_g)>0$. Hence $V(x)<V_g$.
    
    Next, suppose that $V$ exceeds $V_f$. Since
    $V(0)=V_*<V_f$, there must be a maximum at some $x_0\in(0,\ell]$ with $V_f<V(x_0)<V_g$.
    At this point, $V_x(x_0)=0$ and
    $V_{xx}(x_0)\leq 0$, where the latter is understood as a one-sided derivative if $x_0=\ell$. Consequently,
    \[D_{B_L}g(V(x_0))V_{xx}(x_0)\leq 0,\]
    which contradicts $R(V(x_0))>0$ in
    \eqref{eq:outer_nonuc}. Therefore $V(x)\leq V_f$.
    
    Moreover, equality cannot hold at an interior point, since $V=V_f$ would then be a local maximum, and uniqueness of the corresponding initial value problem would imply $V\equiv V_f$. The same uniqueness argument applies at $x=\ell$, where $V_x(\ell)=0$. Thus $V(x)<V_f$ on $[0,\ell]$.
    
    Finally, integrating \eqref{eq:outer_nonuc} from $x$ to $\ell$ yields
    \[D_{B_L}g(V(x))V_x(x)=-\int_x^\ell R(V(s))\,ds.\]
    Since $R(V)<0$ and $g(V)>0$ for $V<V_f$, we conclude that $V_x(x)>0$ for $0<x<\ell$.
    Together with $V(0)=V_*$, this completes the proof.
    \end{proof}
    
    We next examine how the boundary value
    $V_\ell\equiv V(\ell)$ depends on the domain
    half-length $\ell$. Multiplying
    \eqref{eq:outer_nonuc} by $g(V)V_x$ and integrating once gives
    \begin{equation}\label{eq:nonuc_first_integral}
    \frac{D_{B_L}}{2}\big(g(V)V_x\big)^2
    =\mathcal G(V_\ell)-\mathcal G(V),
    \qquad \mathcal G'(V)=-g(V)R(V). 
    \end{equation}
    Since $V_x>0$, separation of variables yields
    \begin{equation}\label{eq:nonuc_length}
    \ell=\chi(V_\ell)\equiv \sqrt{\frac{D_{B_L}}{2}} \int_{V(0^+)}^{V_\ell} \frac{g(\xi)\,d\xi}
    {\sqrt{\mathcal G(V_\ell)-\mathcal G(\xi)}},
    \qquad V_*<V_\ell<V_f.
    \end{equation}
    
    To establish the behavior of $\chi$, under the regime that $V_f<V_g$, so that $g(V)=G'(V)>0$ for $V_*<V<V_f$. Moreover, $f'(V)=2\eta V+g(V)>0,$ and hence $-f(V)$ is positive and strictly decreasing on $(V_*,V_f)$.

    Using $G$ as the integration variable in \eqref{eq:nonuc_length}, with $y=G(V)$ and $dy=g(V)\,dV$, we obtain
    \[ \mathcal G(V_\ell)-\mathcal G(V) = \int_{G(V)}^{G(V_\ell)} z\left[-f\!\left(G^{-1}(z)\right)\right]\,dz. \] Hence,
    \[ \chi(V_\ell) = \sqrt{\frac{D_{B_L}}{2}} \int_{G(V_0^+)}^{G(V_\ell)} \frac{dy}{ \sqrt{\displaystyle \int_y^{G(V_\ell)} z\left[-f\!\left(G^{-1}(z)\right)\right]\, dz}}.\] 
    Rescaling \(y=G(V_\ell)t\) and $z=G(V_\ell)s$ then gives
    \[ \chi(V_\ell)= \sqrt{\frac{D_{B_L}}{2}} \int_{G(V_0^+)/G(V_\ell)}^1 \left[ \int_t^1 s\left[-f\!\left( G^{-1}\!\big(sG(V_\ell)\big) \right)\right]\,ds \right]^{-1/2}dt. \]

    As $V_\ell$ increases, $G(V_\ell)$ increases, so the lower limit decreases. At the same time, $G^{-1}(sG(V_\ell))$ increases for each fixed $s$, and hence $-f(G^{-1}(sG(V_\ell)))$ decreases. It follows that $\chi(V_\ell)$ is strictly increasing.
    
    Furthermore, we have from $f(V_f)=0$ and $f'(V_f)=2\eta V_f+g(V_f)>0$ that
    \[ -f(V)\sim f'(V_f)(V_f-V), \qquad V\to V_f^-. \] 
    Consequently, $\mathcal{G}(V_f)-\mathcal{G}(V)=\mathcal{O}\big((V_f-V)^2\big),$ so the integrand in $\chi(V_\ell)$ given in \eqref{eq:nonuc_length} behaves like $1/(V_f-V)$ near $V_f$. It follows that $\chi(V_\ell)\to\infty$ as $V_\ell\to V_f^-$. Therefore, the outer solution can be continued to arbitrarily large domain lengths, with $V_\ell\to V_f<V_g $ as $\ell\to\infty.$
    Thus, when $V_g>V_f$, the finite-domain loss of solvability associated with spike nucleation does not occur.
    
    

    \section{Spike Self-Replication and Nucleation} \label{sec:rep_nuc}
    We now combine the results in Sec. \ref{sec:inner} and Sec. \ref{sec:outer} to determine the critical domain lengths associated with the two mechanisms. By comparing the corresponding thresholds, we investigate how the dominant spike-generating mechanism depends on the model parameters. We then construct phase diagrams that distinguish regions of self-replication, nucleation, and the absence of nucleation. Finally, we examine the global bifurcation structure of steady-state solutions and compare the numerical results with the asymptotic predictions.

    \subsection{Critical Domain Lengths and Diffusion Thresholds}\label{subsec: Critical Domain Lengths}
    Our asymptotic theory predicts that spike self-replication will occur for the range of the parameter $a$ where \eqref{eq:nas} has a solution with $\mathcal{A} > \mathcal{A}_c \approx 1.347$ but $V_l < V_g$. On such a range for $a$, we predict that for a K-spike solution, spike-replication occurs when
    \begin{equation}\label{eq:thresh_repD}
    D_{B_L} < D_{B_L, K}^{rep} \equiv \frac{2}{K^2 \left[\chi(V_l)\right]^2} \,.
    \end{equation}
    Here, $V_l < V_g$, is calculated numerically from solving \eqref{eq:nas} with $\mathcal{A} = \mathcal{A}_c \approx 1.347$. Equivalently, in terms of the domain half-length $L$, we use $D_{B_L}={D_B/L^2}$ to predict that for a K-spike solution, spike self-replication occurs when
    \begin{equation}\label{eq:thresh_repL}
    L > L_K^{rep} \equiv \sqrt{\frac{D_B}{2}} K \chi(V_l) \, .
    \end{equation}
    Alternatively, spike nucleation near $x = \pm \ell$ is predicted to occur for the range of $a$ where \eqref{eq:first-integral_final} has a solution with $V_l = V_g$ and $\mathcal{A} < \mathcal{A}_c \approx 1.347$. For this range of $a$,  the slope $\mathcal A_{\rm nuc}$ solves
	\begin{equation}
		\mathcal A_{\rm nuc}^2=\frac{2\beta^2}{\mu^2\eta}\Big[\mathcal G(V_g)-\mathcal G\big(V_0^+(\mathcal A_{\rm nuc})\big)\Big],
		\qquad \mathcal A_{\rm nuc}<\mathcal A_c,
		\label{eq:nuc-A}
	\end{equation}  
    we then predict that spike nucleation occurs when
    \begin{equation}\label{eq:thresh_nucD}
    D_{B_L} < D_{B_L, K}^{nuc} \equiv \frac{2}{K^2 \left[\chi_{\max}\right]^2} \, , \quad \mbox{where} \quad   \chi_{\max} \equiv \chi(V_g;\mathcal{A_{\rm nuc}}) \, ,
    \end{equation}
    or, equivalently, when the domain half-length $L$ is sufficiently large so that
    \begin{equation}\label{eq:thresh_nucL}
    L > L_K^{nuc} \equiv \sqrt{\frac{D_B}{2}} K \chi_{\max} \, .
    \end{equation}
	The thresholds in \eqref{eq:thresh_repD}--\eqref{eq:thresh_nucL} are calculated by using Newton's method on the nonlinear algebraic system \eqref{eq:nas}. The critical value of $a$, denoted by $a_c$, that distinguishes whether spike-replication or spike nucleation behavior occurs first as $D_{B_L}$ is decreased, is obtained by numerically solving \eqref{eq:first-integral_final} for $a$, where we set $\mathcal{A} = \mathcal{A}_c\approx 1.347$ and $V_l = V_g$. On the range $0 < a < a_c$ spike self-replication occurs, while spike nucleation occurs on the range $a >a_c$. Since the threshold $a_c$ is independent of $L$, we conclude that either only spike nucleation events or only spike self-replication events will occur as the domain half-length $L$ increases.
    \begin{figure}[ht]
		\centering		\includegraphics[width=1\textwidth]{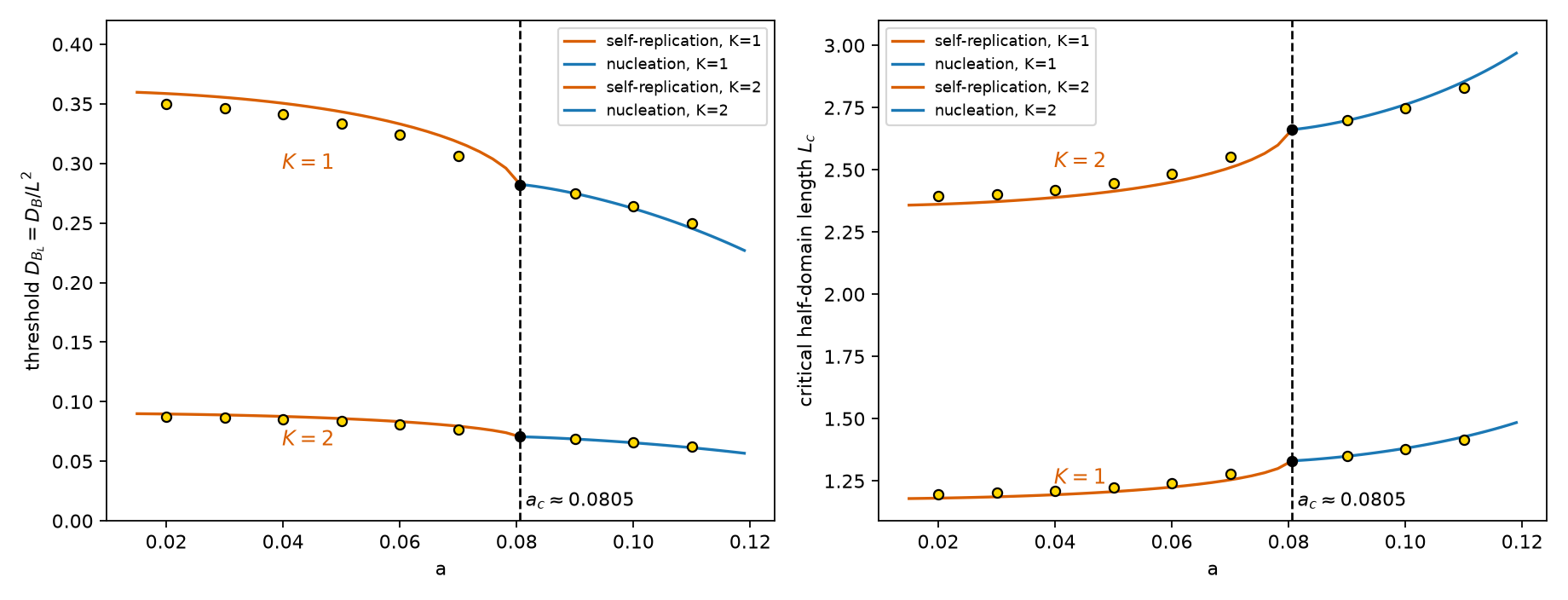}
		\caption{Asymptotic thresholds for spike self-replication $a<a_c$ and nucleation $a>a_c$ for $K=1,2$ spikes, computed from \eqref{eq:thresh_repD} and \eqref{eq:thresh_nucD}  with $D_B=0.5, \varepsilon =0.01$, $\eta=0.5$, and other parameter values in Table~\ref{tab:params}. Left: Threshold $D_{B_L}=D_B/L^2$. Right: Critical domain half-length $L_c$. The black dot marks the critical value $a_c\approx 0.08$. The yellow dots represent results from full numerical simulations using \textit{FlexPDE} \cite{flexpde2015}.
		}
		\label{fig:nuc_rep_thresholds}
	\end{figure}
    For $\eta=0.5, \varepsilon = 0.01$ and $D_B = 0.5$, $K = 1, 2$, and other parameter values given in table \ref{tab:params},  in Fig.~\ref{fig:nuc_rep_thresholds} we plot the spike self-replication threshold $D_{B_L, K}^{rep}$ and the spike nucleation threshold $D_{B_L, K}^{nuc}$ versus $a$, as computed from \eqref{eq:thresh_repD} and \eqref{eq:thresh_nucD}, respectively, as $a$ is varied. The computed threshold $a = a_c \approx 0.08$ separates whether spike self-replication or spike nucleation occurs. We conclude that, as $L$ increases, which corresponds to decreasing $D_{B_L} = D_B/L^2$ and $\varepsilon_L = \varepsilon/L$, spike self-replication owing to the non-existence of the core solution is the dominant instability when $a < a_c$. In contrast, on the range $ a >a_c$, spike nucleation events owing to the non-existence of the outer solution are predicted to be the dominant instability as $D_L$ is decreased. As indicated by the yellow dot marks in Fig.~\ref{fig:nuc_rep_thresholds}, we show a very close agreement between our prediction of the two stability thresholds with the corresponding results computed from full numerical time-dependent simulations of \eqref{TBV_main} using \textit{FlexPDE} \cite{flexpde2015}. The initial conditions for the full  simulations in Fig.~\ref{fig:nuc_rep_thresholds} is a K-spike equilibrium with $K=1$ and $K=2$ respectively. Then numerical values of $D_{B_L}$ were recorded as instability occurs.  Setting $D_B = 0.5$, in the rigth panel of Fig.~\ref{fig:nuc_rep_thresholds} we show a similar plot but in terms of the two critical length thresholds $L_K^{rep}$ and $L_K^{nuc}$ as defined in \eqref{eq:thresh_repL} and \eqref{eq:thresh_nucL}, respectively.
    \subsection{Phase Diagrams and Competing Mechanisms}\label{subsec:phase_diagram}
    
    We now combine the critical thresholds derived in Section \ref{subsec: Critical Domain Lengths} to construct phase diagrams illustrating the competition between spike self-replication and nucleation. In particular, we consider the self-replication and nucleation thresholds, together with the boundary $V_g=V_f$, which separates parameter regimes where nucleation is possible from those where it is excluded. We also include the Turing bifurcation curve, which marks the onset of spatial pattern formation from the homogeneous steady state. Together, these curves partition the parameter space into distinct regions characterized by the existence and selection of different pattern-forming mechanisms.
    
    In regions where spatially localized patterns exist, the relative sizes of the critical domain lengths $L_{\mathrm{rep}}$ and $L_{\mathrm{nuc}}$ determine which spike-generating mechanism is expected to occur first as the domain grows. More specifically, when $L_{\mathrm{rep}}<L_{\mathrm{nuc}}$, spike self-replication occurs first as the domain grows, whereas if $L_{\mathrm{nuc}}<L_{\mathrm{rep}}$, spike nucleation occurs first. In the regime $V_g>V_f$, spike nucleation does not occur, although self-replication may still occur if the critical saddle-node threshold of the inner core problem is reached. In this way, the phase diagrams provide a summary of both the onset of pattern formation and the subsequent competition between spike self-replication
    and nucleation.
    \begin{figure}[ht]
		\centering
		\includegraphics[width=0.8\textwidth]{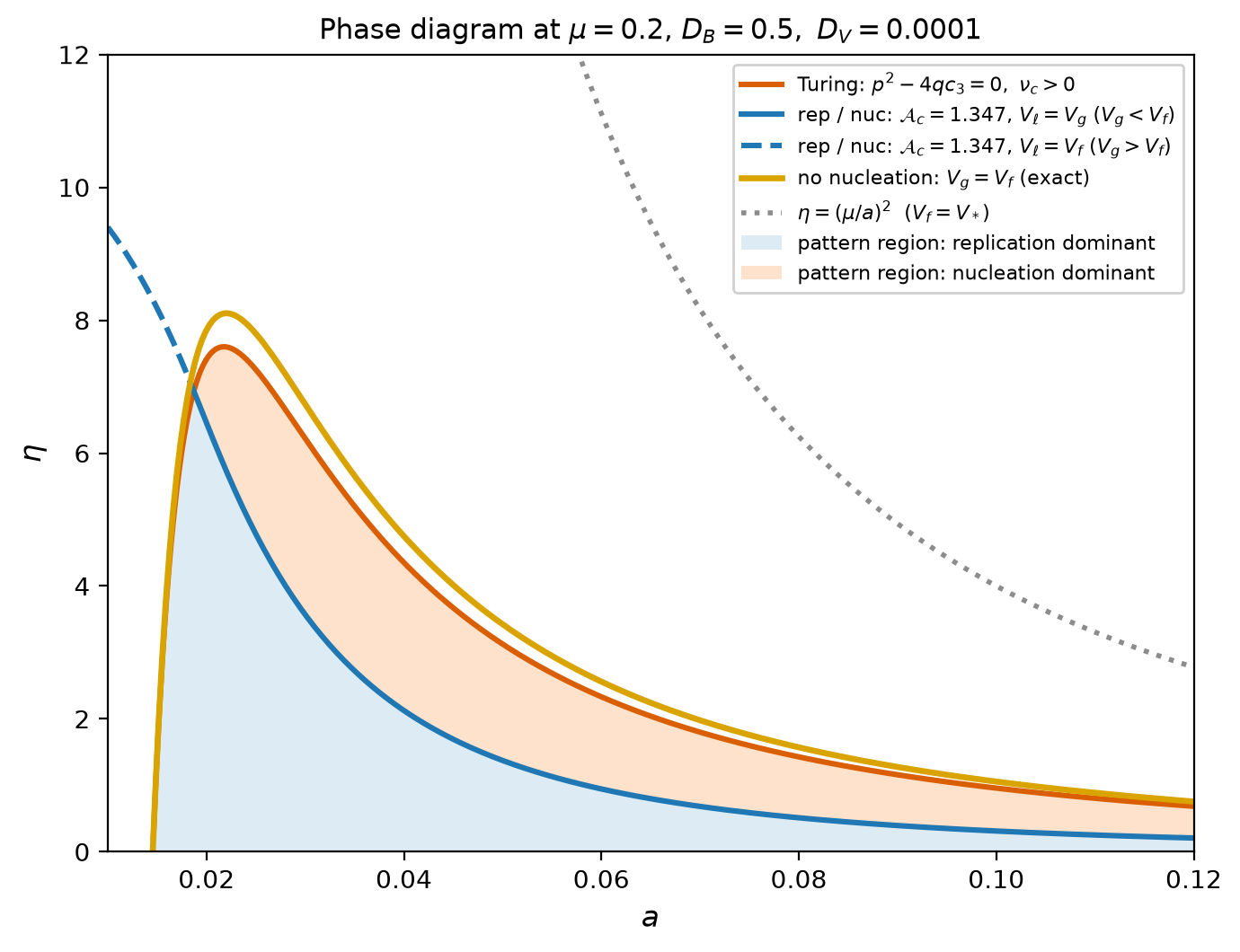}
		\caption{Phase diagram of the TBV model in the $(a,\eta)$ plane for  $\mu=0.2$,  $D_B=0.5$, $D_V=\varepsilon^2=10^{-4}$ and other parameter values given in table \ref{tab:params}. Orange: Turing threshold given in \eqref{eq:stationary_threshold}; the homogeneous state is Turing unstable below this curve. Blue: replication-nucleation boundary obtained from $\mathcal A=\mathcal A_c\approx1.347$ with $V_\ell=V_g$ where $V_g<V_f$ (solid) and $V_\ell=V_f$ where $V_g>V_f$ (dashed). Yellow: no-nucleation curve $V_g=V_f$; nucleation is possible only below it. Grey dotted: $\eta=(\mu/a)^2$, at which $V_f=V_*$. The shaded area is the pattern region below the Turing curve, split by the blue curve into a replication-dominant part (light blue) and a nucleation-dominant part (light orange).}
		\label{fig:phase_diagram}
	\end{figure}

    Fig.~\ref{fig:phase_diagram} shows the phase diagram for $\mu=0.2$. The Turing curve $p^2-4qc_3=0$ (orange) marks the onset of Turing instability, with the homogeneous state unstable below it. Within this region, the replication-nucleation curve (blue) separates two spike-generating mechanisms. Below the blue curve (light blue), the core parameter $\mathcal A$ reaches its critical value $\mathcal A_c$ before the outer solution reaches the fold $V_g$, leading to spike self-replication as $\ell$ increases. Between the blue and orange curves (light orange), nucleation occurs first. 
    
    The no-nucleation curve $V_g=V_f$ (yellow) lies above the Turing curve, indicating that nucleation is possible throughout the Turing-unstable region. Above the no-nucleation curve, the outer solution approaches $V_f$ without reaching $V_g$, and hence nucleation cannot occur. For sufficiently small $a$, the replication-nucleation curve crosses above the rising branch of the Turing curve, so the entire Turing-unstable region is dominated by self-replication. In the regime $V_g>V_f$ (dashed blue), the replication threshold is instead evaluated at $V_\ell=V_f$, corresponding to the large-domain limit of the outer solution.
    
    The ordering of the Turing and no-nucleation curves is not specific to the parameter values used in the figure, but follows from the structure of the model. At the homogeneous steady state, $\bar B=G(\bar V)=1-\eta\bar V^2$, so that $f(\bar V)=0$ and $\bar V=V_f$. Using \eqref{eq:pq}, 
    we obtain 
    \begin{equation}\label{eq:p_Gprime}
    p=-f_T\left[\frac{D_B\beta\bar V^2}{1+\bar T}G'(V_f)+D_V\bar B \right].
    \end{equation}
    Since $f_T<0$ and $\bar B>0$, the Turing condition $p<0$ requires $G'(V_f)<0$. Recall that $G'=g$, where $g$ has a unique positive root $V_g$ and changes sign from positive to negative at $V_g$. Thus, $G$ increases for $V<V_g$ and decreases for $V>V_g$. Consequently, $G'(V_f)<0$ implies $V_f>V_g$, which is precisely the regime where spike nucleation is possible. This explains why the Turing curve lies below the no-nucleation curve.
    
    The two curves become asymptotically close as $D_V=\varepsilon^2\to0$. Indeed, at the Turing threshold, $p=-2\sqrt{qc_3}$, where $q=-f_TD_B\varepsilon^2$. Substituting into \eqref{eq:p_Gprime} gives, to leading order,
    \[G'(V_f)\sim -\frac{2(1+\bar T)}{\beta\bar V^2} \sqrt{\frac{c_3}{-f_T}}\, \frac{\varepsilon}{\sqrt{D_B}}.\]
    Since $G'(V_g)=0$ and $G''(V_g)<0$, it follows that $V_f-V_g=O\left(\frac{\varepsilon}{\sqrt{D_B}}\right).$
    
    In Fig. \ref{fig:phase_diagram}, the no-nucleation curve is computed using the exact first root $V_f$ of $f$. The condition $V_g=V_f$ gives $\eta_{\rm nn}=\frac{1-G(V_g)}{V_g^2}.$ The asymptotic approximation \eqref{eq:Vf-approx} requires $G(V_0)\ll1$ and becomes less accurate as $a$ decreases or $\mu$ increases, since $G(V_g)\approx \frac{\mu^2(1+h(V_g))}{4\beta a}.$
    
    For smaller $\mu$, the discrepancy is noticeable only at small $a$. For larger $\mu$, $G(V_g)$ approaches one, and the exact no-nucleation curve reaches $\eta_{\rm nn}=0$. When $a$ decreases further, nucleation is absent for all $\eta$. This behavior is not captured by the asymptotic approximation.
    
    Finally, the localized-pattern thresholds are computed using the leading-order approximation $T\simeq h(V)=\alpha/[\delta(1+V)]$, which neglects the nonlinear self-enhancement term and introduces a relative error bounded by $e=\rho\alpha/(\delta^2k^2)\ll1$. Whereas the Turing curve is obtained from the exact homogeneous steady state. Recomputing the localized-pattern thresholds using the exact relation for $T$ produces only $O(e)$ corrections, corresponding to a small shift in $a$ along the steep left branches.

    \subsection{Global bifurcation diagrams}\label{sch:global}

\begin{figure}[htbp]
    \centering
    \includegraphics[width=0.99\textwidth, height=6.0cm]{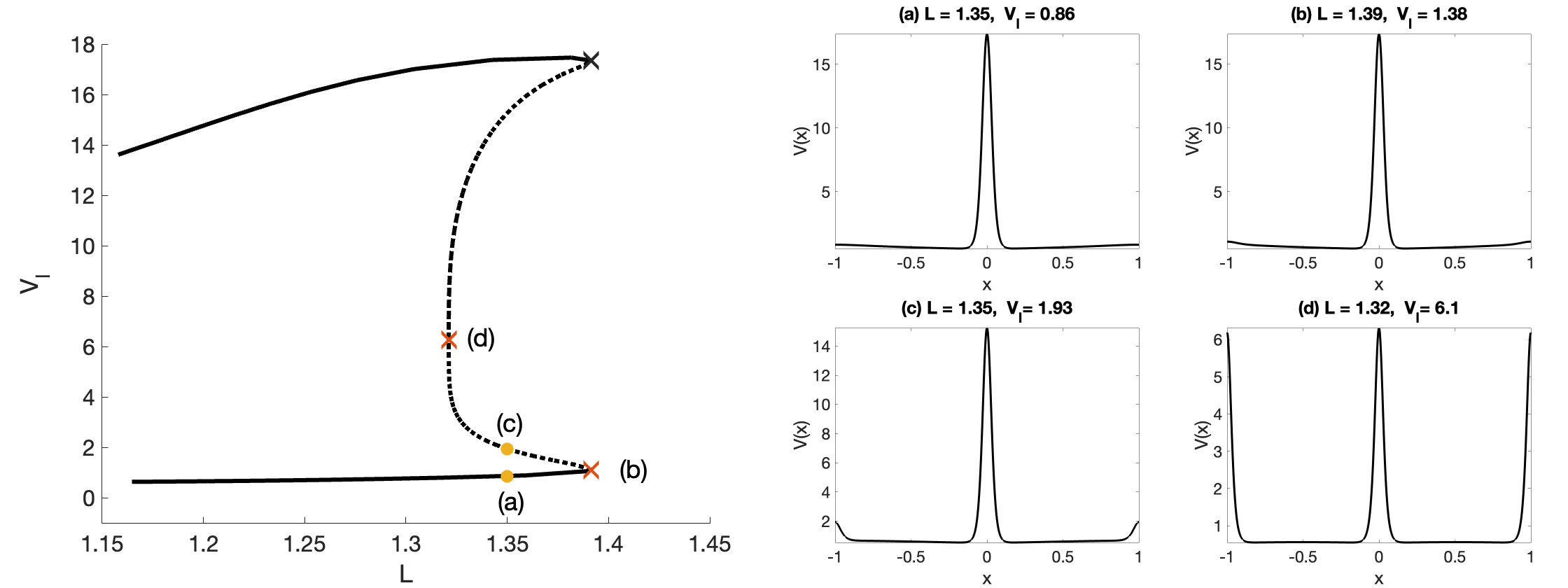}
    \caption{Left panel: Global bifurcation diagram of single-spike steady-states for the TBV model \eqref{TBV_main} as computed using \textit{pde2path} \cite{pde2path}, where $L$ is the bifurcation parameter and the vertical axis is $\mu = v(\ell)$. Parameters: $a = 0.1$, other parameters are the same as in Fig. \ref{fig:nuc_rep_thresholds}. Since $a>a_c $, where $a_c \approx 0.08$, we predict that spike nucleation behavior occurs as $L$ is increased starting from point (a) on the lower branch. The red point (b) marks a saddle-node point that is connected to an unstable middle branch that has an interior spike together with two emerging boundary spikes. Solid portions are linearly stable while dashed portions are unstable. Right panel: Spike profile $V(x)$ and bifurcation values (top of subfigures) at the indicated points in the left panel.}
    \label{fig:bif_nuc}
\end{figure}

In this section, we numerically compute global bifurcation diagrams of one-spike steady-state solutions for \eqref{TBV_main} using \textit{pde2path} \cite{pde2path}, where $L$ is taken as the bifurcation parameter and we fix other parameters given in Fig. \ref{fig:nuc_rep_thresholds}. When $a = 0.1$ is chosen on the range $ a >a_c$, where $a_c \approx 0.08$, the global bifurcation diagram of single-spike steady-states shown in Fig.~\ref{fig:bif_nuc} provides a detailed view on how new spikes are created near the domain boundaries when $L$ approaches the saddle-node bifurcation point, which we denote by $L^{num}$. In this bifurcation diagram, the vertical axis in the left panel of Fig.~\ref{fig:nuc_rep_thresholds} is the boundary value $V_l\equiv V(\ell)$. Solid and dashed portions of the branches indicate linearly stable and unstable solutions, respectively. In the right panel of Fig.~\ref{fig:nuc_rep_thresholds} the solution $V(x)$ at the marked points is shown. Starting from the bottom solution branch at point (a), this branch that has one interior spike ceases to exist as $L$ increases above $L^{num}$, where $L^{num} \approx 1.39$. We emphasize that this saddle-node point is well-approximated by the value $L_1^{nuc}$ in \eqref{eq:thresh_nucL}, at which the outer solution ceases to exist (this is also confirmed by Fig.~\ref{fig:nuc_rep_thresholds} for $a = 0.1$ and $K = 1$). Traversing the middle branch on Fig.~\ref{fig:bif_nuc} starting from the saddle-node point, we observe at point (c) that new boundary spikes are nucleated at the domain endpoints.

In Fig.~\ref{fig:bif_rep} we use \textit{pde2path} \cite{pde2path} to obtain a global bifurcation diagram for single-spike steady-states of \eqref{TBV_main} for the value $a = 0.05 < a_c$, where we predict that spike self-replication will occur as $L$ is increased. In the left panel of Fig.~\ref{fig:bif_rep}, the vertical axis is now $V(0) B(0)$. The computed saddle-node value is $L^{num} \approx 1.22$, which is well approximated by the critical threshold $L_1^{rep}$ in \eqref{eq:thresh_repL}, for which the core problem \eqref{eq:core-canonical} has a saddle-node point when coupled to the outer solution (see Fig.~\ref{fig:nuc_rep_thresholds} for $a = 0.05$ and $K = 1$).

\begin{figure}[htbp]
    \centering
    \includegraphics[width=0.99\textwidth, height=6.0cm]{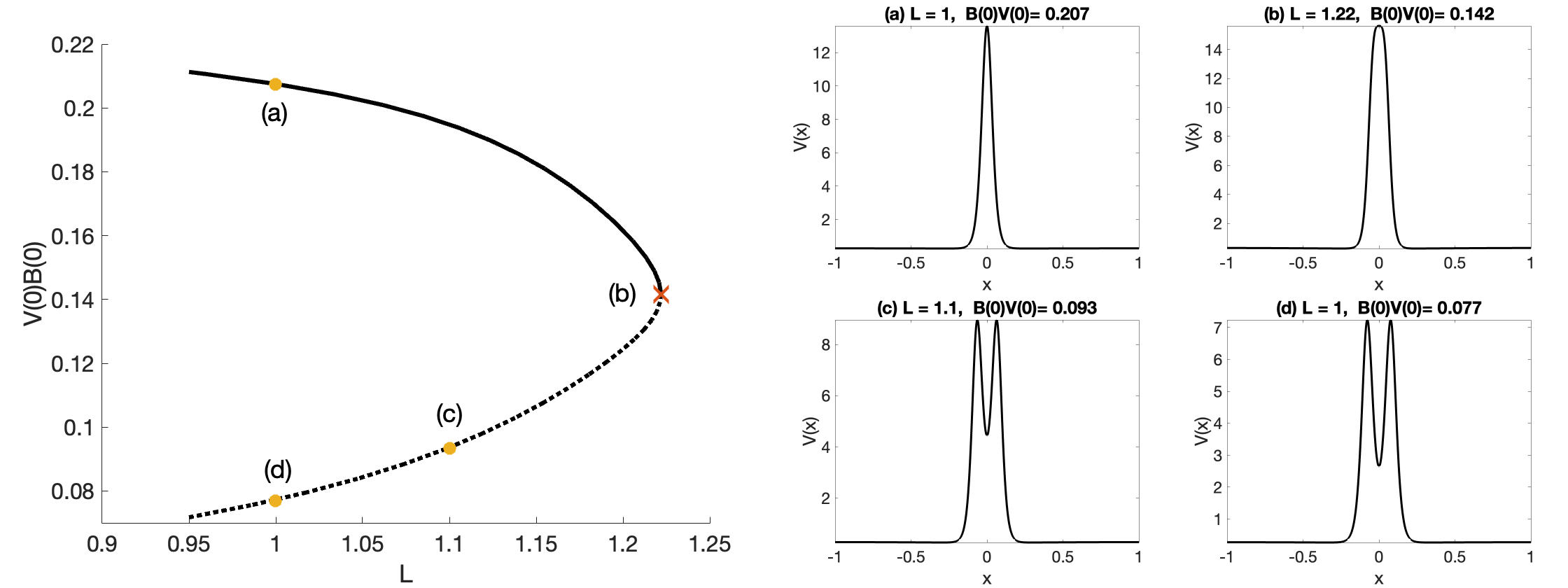}
    \caption{Left panel: Same plot and parameters as in Fig.~\ref{fig:bif_nuc} except that $a = 0.05 < a_c$. The vertical axis is now $V(0) B(0)$, which better measures the spike amplitude. Since $a < a_c$, we predict that spike self-replication will occur as $L$ is increased starting from point (a) on the upper branch. The red star marks a saddle-node point that is a signature for spike self-replication. The linearly stable upper branch is connected to an unstable lower branch where the spike profile has a volcano shape. Observe that the global bifurcation diagram has a qualitatively similar form to that for the core problem in Fig.~\ref{fig:core_bif}. Right panel: Spike profile $V(x)$ and bifurcation values (top of subfigures) at the indicated points in the left panel.}
    \label{fig:bif_rep}
\end{figure}

    \section{Discussion} \label{sec:discussion}
    In this paper, we have investigated pattern formation in a reaction--diffusion model for the tulip breaking virus (TBV), with particular emphasis on localized patterns and their evolution on growing domains. Previous work by Wong et al. \cite{wong2025tulip} introduced the TBV model and demonstrated how the interaction between viral infection, resource availability, and pigment production can generate striped patterns on growing tulip petals. Building on this work, we have carried out a more detailed analysis of the underlying pattern-forming mechanisms. In particular, our study connects the classical Turing instability of the homogeneous steady state with the existence of strongly localized spike patterns and the mechanisms by which new spikes are generated as the domain grows.
    
    In the regime of small viral diffusivity, we used matched asymptotic expansions to construct localized spike solutions and derive the critical conditions for spike self-replication and nucleation. These two mechanisms are governed by different parts of the solution: self-replication is associated with a saddle-node threshold of the inner spike core, whereas nucleation occurs when the outer solution reaches a fold of the nonlinear nullcline. By combining these thresholds with the Turing instability condition, we obtained a unified phase diagram that identifies the parameter regimes associated with different pattern-generating mechanisms. The phase diagram also identifies a regime in which spike nucleation does not occur, although self-replication may remain possible. These results extend previous asymptotic studies of spike-generating mechanisms in classical reaction--diffusion systems~\cite{gai2025asymptotic} to the biologically motivated TBV model.
    
    An important finding of this study is the analytical connection between Turing instability and spike nucleation. We showed that, for the TBV model, the Turing-unstable region lies entirely within the regime $V_g<V_f$, where spike nucleation is possible. This ordering follows from a relation between the Turing instability condition and the slope of the outer nullcline, and is independent of the particular parameter values considered. Moreover, the Turing and no-nucleation boundaries become asymptotically close as the ratio of viral to resource diffusivity tends to zero. Nevertheless, the condition $V_g<V_f$ only allows nucleation; whether nucleation or self-replication occurs first depends on their respective critical thresholds.
    
    While previous asymptotic studies have established phase diagrams for spike self-replication and nucleation in several reaction--diffusion models \cite{gai2025asymptotic}, their relationship with Turing instability was not considered. This connection is not immediate, since Turing instability concerns the linear stability of the homogeneous state, whereas spike nucleation is determined by the nonlinear outer problem. Related links between Turing bifurcation and the emergence of localized states have been studied in activator--inhibitor systems \cite{villar2023degenerate}.  It would therefore be interesting to investigate whether similar relationships arise in other reaction-diffusion models and to identify the conditions under which they hold.
    
    From a biological perspective, the localized viral spikes represent regions of elevated viral concentration, where pigment production is reduced and resources are depleted \cite{lesnaw2000tulip}. The emergence of additional spikes during petal growth therefore provides a possible mechanism for the formation of new regions of reduced pigmentation, contributing to the development of striped patterns \cite{wong2025tulip}. Our results further suggest that the spatial distribution of these regions may depend on the relative rates of viral replication, resource depletion, and diffusion. In particular, changes in these parameters can alter whether existing viral regions split or new regions develop between them. Thus, the analysis provides a mechanistic interpretation of how different stripe patterns may arise during petal growth.  Moreover, the constant viral input $a$ also plays an important role in determining the pattern-generating mechanism. By maintaining a nonzero viral background, it modifies the outer solution and can shift the balance between spike self-replication and nucleation.
    
    Several questions remain open. The critical thresholds derived here assume sufficiently slow domain growth. It would be interesting to investigate how finite growth rates affect the onset of spike self-replication and nucleation, including possible delayed transitions. Further investigation of the global bifurcation structure, together with the stability and interactions of multiple localized spikes, may also provide insight into the selection and spacing of pigmentation patterns during petal growth. These studies could help explain the diversity of stripe patterns observed in virus-infected tulips.

    \section*{Acknowledgements}
    Chunyi Gai is grateful for the support of the NSERC Discovery grant program. Fanya Xia was supported by the Mitacs Globalink Research Internship program.
    
	\bibliographystyle{plain} 
    \bibliography{refs}

\end{document}